\documentclass[journal]{IEEEtran}
\usepackage{amsmath,amssymb,amsfonts,amsthm}
\usepackage{amsmath}
\usepackage[caption=false,font=footnotesize,labelfont=rm,textfont=rm]{subfig}
\usepackage{graphicx}
\usepackage[dvipsnames, svgnames, x11names]{xcolor}
\usepackage{bbm}

\usepackage[colorlinks=true,linkcolor=red,citecolor=green]{hyperref}
\newtheorem{Definition}{Definition}
\newtheorem{Proposition}{Proposition}
\newtheorem{lemma}{Lemma}
\usepackage{algorithmic,algorithm}
\usepackage{setspace}
\renewcommand{\Re}{\operatorname{Re}}

\usepackage{float}
\usepackage{subfig}
\usepackage{overpic}
\usepackage{epstopdf}
\usepackage{cases}
\usepackage{bm}

\usepackage{stfloats} 

\allowdisplaybreaks[4]

\begin{document}
\title{Age of Semantic Information-Aware Wireless Transmission for Remote Monitoring Systems}
\author{	
\IEEEauthorblockN{Xue Han, Biqian Feng, Yongpeng Wu, \emph{Senior Member, IEEE}, Xiang-Gen Xia, \emph{Fellow, IEEE}, \\
Wenjun Zhang, \emph{Fellow, IEEE}, and Shengli Sun}
\thanks{X. Han, B. Feng, Y. Wu, and W. Zhang are with the Department of Electronic Engineering, Shanghai Jiao Tong University, Minhang 200240, China (e-mail: han.xue@sjtu.edu.cn; fengbiqian@sjtu.edu.cn; yongpeng.wu@sjtu.edu.cn;  zhangwenjun@sjtu.edu.cn) (Corresponding author: Yongpeng Wu).}

\thanks{Xiang-Gen Xia is with the Department of Electrical and Computer Engineering, University of Delaware, Newark, DE 19716 USA (e-mail: xxia@ee.udel.edu).}

\thanks{Shengli Sun is with Shanghai Institute of Technical
Physics, Chinese Academy of Sciences, Shanghai 200083, China (e-mail: Palm\_sun@mail.sitp.ac.cn).}
}
\maketitle

\begin{abstract}
Semantic communication is emerging as an effective means of facilitating intelligent and context-aware communication for next-generation communication systems. In this paper, we propose a novel metric called Age of Incorrect Semantics (AoIS) for the transmission of video frames over multiple-input multiple-output (MIMO) channels in a monitoring system. Different from the conventional age-based approaches, we jointly consider the information freshness and the semantic importance, and then formulate a time-averaged AoIS minimization problem by jointly optimizing the semantic actuation indicator, transceiver beamformer, and the semantic symbol design. We first transform the original problem into a low-complexity  problem via the Lyapunov optimization. Then, we decompose the transformed problem into multiple subproblems and adopt the alternative optimization (AO) method to solve each subproblem. Specifically, we propose two efficient algorithms, i.e., the successive convex approximation (SCA) algorithm and the low-complexity zero-forcing (ZF) algorithm for optimizing transceiver beamformer. We adopt exhaustive search methods to solve the semantic actuation policy indicator optimization problem and the transmitted semantic symbol design problem.
Experimental results demonstrate that our scheme can preserve more than 50\% of the original information under the same AoIS compared to the constrained baselines.

\end{abstract}


\section{Introduction}
With the development towards the sixth generation (6G) of wireless communication networks, a widely accepted challenge is data transmission within these exponentially growing applications, e.g., digital twins, autonomous vehicles, and intelligent environments, which poses tremendous challenges to the utilization of limited bandwidth resources \cite{Zhu}.
Semantic communication, as an emerging communication paradigm, is envisioned to be a potential solution to these systems requiring extremely high data rates and ultra-low latency, which facilitates the seamless integration of information and communication technology with artificial  intelligence (AI) \cite{Yang_app, Han_SCSC}. Different from traditional communication methods, semantic communication emphasizes the transmission of valuable semantic-aware information rather than accurate bit recovery, thereby guaranteeing improved transmission efficiency and reliability \cite{Gunduz_bit}.

The fundamental idea underlying semantic communication is to connect the source and channel components of Shannon's theory \cite{Shannon_theory}, thereby facilitating end-to-end wireless transmission. Recently, the rapid development of deep learning (DL) has inspired research interest in deep joint source and channel coding (JSCC) to enable the efficient transmission of semantic communication for various data types over wireless channels \cite{Xie_JSCC}-\cite{Wang_video}. For example, 
Bourtsoulatze et al. \cite{Bourtsoulatze_img} proposed a novel neural network that combines source and channel coding for image compression and transmission.
Wang et al. \cite{Wang_video} proposed a deep video semantic transmission structure based on the nonlinear transform and conditional coding to adaptively extract semantic features according to the source entropy.

To improve the performance of semantic communication, especially in the realm of deep JSCC, researchers have introduced various new metrics and delved deeply into semantic resource allocation. Yan et al. \cite{Yan_text} proposed semantic spectral efficiency in an uplink scenario and optimized resource allocation for text-based semantic communication. Furthermore, the work in \cite{Yan_QoE} defined semantic entropy for different tasks and investigated a quality-of-experience (QoE) maximization problem in terms of the number of transmitted semantic symbols, channel assignment, and power allocation. 
For the exploration of efficient resource allocation strategies \cite{Mu_NOMA}-\cite{Yang_22}, for example, Mu et al. \cite{Mu_NOMA} investigated the resource allocation challenges within the downlink non-orthogonal multiple access (NOMA) system. Chi et al. \cite{Chi_ICC} focused on optimizing compression ratio, power allocation, and resource block assignment to enhance user numbers and image quality. 
Yang et al. \cite{yang_2023} investigated wireless resource allocation and semantic extraction for energy-efficient communications using rate splitting, aiming to minimize total communication and computation energy consumption. 
Xia et al. \cite{xia_2023} addressed user association and bandwidth allocation challenges through the development of new metrics and optimization solutions, leading to improved network performance in both perfect and imperfect knowledge-matching scenarios. 
Yang et al. proposed a Bayesian hybrid beamforming \cite{Yang_11} and developed a Pareto-optimal resource allocation scheme for energy-efficient transmission  \cite{Yang_22}.
Specifically, a stale, incorrect, or energy-inefficient delivery would aggravate the loss of utility of information. Despite the positive outcomes demonstrated by these approaches, the aspect of information freshness within the JSCC framework remains under-explored. 

The freshness of information, to a certain extent, reflects the significance and importance of information. A well-known metric to  measure the importance and priority of information carried by data is the age of information (AoI), which is defined as the time elapsed since the generation of the latest received update \cite{Kaul_AoI}\cite{Liaimin_Goal}. 
Recent studies have proposed various different evaluation metrics, such as Age of Incorrect Information (AoII) \cite{AoII}, Urgency of Information \cite{UoI}, Value of Information (VoI) \cite{VoI}, Query AoI \cite{QoI}, and unified closed-form average AoI \cite{Liaimin_ARQ}, which compensate for the limitation of AoI in failing to capture the content of information. The AoII is proposed to extend the notion of fresh updates to that of fresh informative updates.
Specifically, for semantic-aware scenarios, Maatouk et al. \cite{Maatouk_2022_AoI} introduced the AoII to enhance semantics-empowered communication and proposed an optimal randomized threshold transmission strategy. 
Agheli et al. \cite{Agheli_2022_AoI} proposed semantics of information (SoI) to measure the significance and usefulness of semantic information concerning the goal of data exchange. 
This work is extended into distributed wireless monitoring systems with multiple sensors and multiple monitors \cite{Agheli_2024_AoI}.
Luo et al. \cite{Luo_2024_AoI} investigated semantic-aware communication for remote estimation of multiple Markov sources and developed efficient algorithms to overcome dimensionality issues by leveraging information semantics.
Chen et al. \cite{Chen_2024_AoI} introduced the Age of Semantic Importance (AoSI) metric and proposed an algorithm based on Deep Q-Network to optimize multi-source scheduling and resource allocation.
Moreover, Hu et al. \cite{Hu_2024_AoI} introduced the concept of Version Age of Information (VAoI) to federated learning and integrated timeliness and content staleness into client scheduling policies to minimize average VAoI.
However, it is worth mentioning that although these pioneering works demonstrate promising prospects and performance results for integrating AoI with semantic communication, their adoption in practical monitoring systems is challenging. 

For application-centric approaches, especially within the deep JSCC framework, there are few studies on semantic-aware scheduling and resource allocation considering both the information freshness and semantic importance for video transmission. 
When it comes to monitoring systems, integrating AoI with semantic communication faces several challenges. On one hand, monitoring systems usually have strict requirements for data freshness to ensure timely and accurate monitoring results. On the other hand, the semantic information processing in semantic communication needs to consider the semantic importance of video content, which may conflict with the traditional AoI-based resource allocation strategies that focus on information freshness. Furthermore, the resource-constrained wireless environment adds another layer of difficulty. Under such circumstances, it is challenging to design an efficient resource allocation strategy and semantic information processing method that can meet the diverse demands of practical monitoring systems, such as the wildland fire alarm system, which requires a high level of data freshness but is less concerned with data perfect pixel-level reconstruction. 

This motivates us to explore an efficient end-to-end wireless transmission and resource optimization scheme to achieve effective data transmission while maintaining the information freshness. Moreover, the limited computing resources and scarce spectrum resources of the wireless communication system lead to intense competition for resources, especially during peak times. Semantic communications focus on the purpose of transmitting meanings, instead of precisely recovering all original bits, and therefore improves spectral efficiency.

In this work, we go one step further and investigate the age of semantic information in wireless video transmission and resource allocation, focusing on both content and freshness for remote monitoring systems. We aim to construct an efficient joint optimization strategy for the semantic actuation indicator, transceiver beamformer, and semantic extraction to minimize the time-averaged age of incorrect semantics (AoIS) of all users in a video monitoring wireless network. 
Specifically, the base station first determines whether to perform the semantic information extraction and transmission. 
Subsequently, we formulate a time-averaged AoIS optimization problem by considering the semantic actuation indicator, transmit and receive beamformer, and the semantic symbol design simultaneously. This joint optimization problem is spatially and temporally coupled in different time slots, leading to a high-dimensional mixed-integer dynamic program that is difficult to solve in practice. To solve this stochastic optimization problem, we first adopt the Lyapunov optimization to transform it into a per-slot real-time optimization problem with multiple subproblems, and then employ alternating optimization (AO) to solve the subproblems.
The main contributions of this paper can be summarized as follows. 
\begin{itemize}
\item {\textit{AoIS-aware Semantic Communication:} 
We consider a wireless video semantic transmission framework in which only important frame information is transmitted, greatly reducing the communication overhead. Then, we propose the concept of AoIS and introduce an indicator for semantic actuation policy design. 
}

\item {\textit{Lyapunov Optimization for the AoIS Minimization:}  We adopt AoIS to quantify information freshness and semantic information value. We formulate a long-term AoIS minimization problem by jointly optimizing the actuation indicator, transceiver beamformer, and semantic extraction. Then, the time-averaged AoIS minimization problem is transformed into a queue stability problem according to the Lyapunov optimization method.}

\item {\textit{Problem Solution:} We employ the AO method to optimize all variables alternately. Specifically, we propose a successive convex approximation (SCA)-based algorithm to solve the transceiver beamformer design problem in the multiple-input multiple-output (MIMO) scenario, and propose a low-complexity zero-forcing (ZF) algorithm to optimize the transceiver beamformer in the multi-user multiple-input single-output (MISO) system. The exhaustive search method is adopted to optimize both the semantic actuation indicator and the semantic symbol design problems.}

\item {\textit{Performance Validation:} Simulation results demonstrate the convergence of our proposed algorithms. The AoIS-aware algorithms demonstrate that semantic extraction significantly reduces AoI while maintaining the high quality of the reconstructed frames, highlighting the potential of these methods for improving real-time video transmission in resource-constrained environments. }
\end{itemize}

The remainder of the paper is structured as follows. Section \ref{sec:system model} introduces the considered system model, the definition of the novel metric, AoIS, and formulates an optimization problem. Section \ref{sec:perfect} describes the Lyapunov transformation and several algorithms for solving the problem. Section \ref{sec:LCZF} introduces a low-complexity ZF transceiver beamforming optimization algorithm. Section \ref{sec:Results} presents the numerical results and performance evaluation. Finally, Section \ref{sec:Conclusion} concludes this paper.

\textit{Notations}:  We adopt $x$, $\mathbf x$, and $\mathbf X$ to denote a scalar, vector, and matrix, respectively. Superscripts $T$ and $H$ stand for the transpose and conjugate transpose, respectively. $ \mathbb{R} $ and $\mathbb{C} $ refer to the real and complex number sets, respectively. $\| \cdot \|$ denotes the Frobenius norm and $\mathbb E$ is the mathematical expectation. $\mathcal{CN}(\mu,\sigma^2)$ represents the circularly symmetric complex Gaussian distribution with mean $\mu$ and variance $\sigma^2$. $\odot$ is the Hadamard product.

\section{System Model}
\label{sec:system model}
In this section, we describe the wireless video downlink transmission in remote monitoring scenarios and the considered semantic communication framework. Then, we introduce the definition of the proposed AoIS metric and the design of the semantic actuation efficiency policy. Moreover, an optimization problem is formulated to minimize the total actuation cost for the semantic video frame transmission.   

\subsection{Semantic Communication-Empowered Remote Monitoring System}
We consider a remote monitoring system comprising $U$ device users, each equipped with $N_r$ antennas. Given the raw video $\mathbf S$ sent by the base station (BS) with $N_t$ antennas, it is processed to a semantic sequence before being transmitted to the device users and can be reconstructed as $\hat{\mathbf S}$ by their video decoder.

\subsubsection{Semantic-aware Transmitter}

\begin{figure*}[htbp]
\centering
\includegraphics[width=0.9\linewidth]{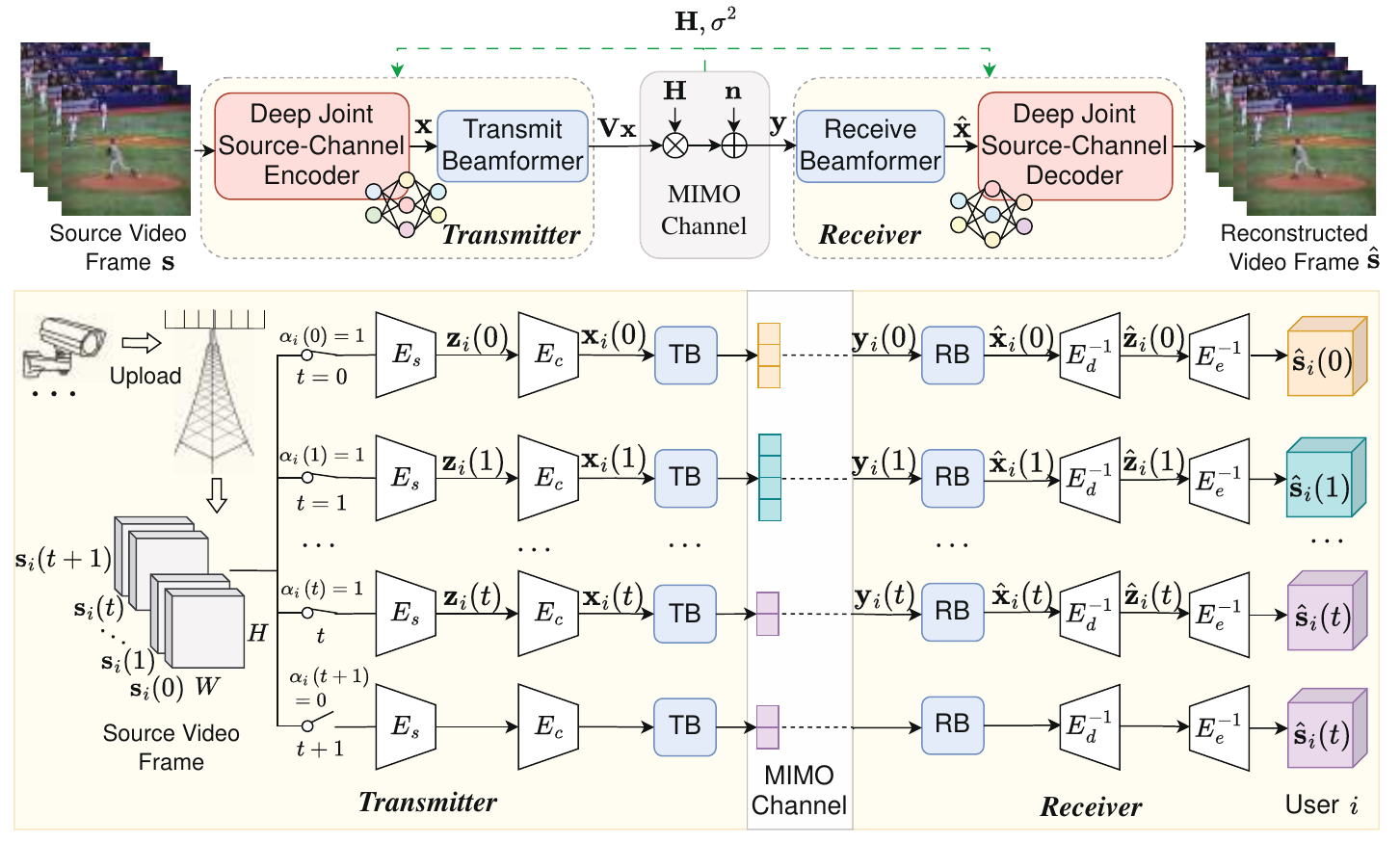}
\caption{System model of the AoIS-aware multi-user semantic communication network in the remote monitoring system.}
\label{system_2}
\end{figure*}

As shown in Fig. \ref{system_2}, the BS aims to transmit semantic information extracted from the captured source video frame $\{\mathbf s_i(0), \mathbf s_i(1), \cdots\}$ for device user $i$ in real time, where $\mathbf {s}_{i} (t) \in \mathbb{R}^{H \times W \times 3}$ denotes the $t$-th source video frame. Specifically, a semantic encoder, $E_{\varrho}(\cdot): \mathbb{R}^{H \times W \times 3} \mapsto \mathbb{R}^{N_H  \times 1}$ parameterized
by the offline-trained weight $\varrho$, is adopted to encode the raw video frame $\mathbf s_i (t)$ into the semantic features $\mathbf z_i (t)\in \mathbb R^{N_H \times 1} $, where $N_H$ refers to the sequence length of the extracted semantic information.  Subsequently, the variable-length channel encoder, $E_{\nu}(\cdot): \mathbb{R}^{{N_H} \times 1} \mapsto \mathbb{C}^{L_i(t) \times 1}$, maps the aforementioned feature ${\mathbf z}_i (t)$ into the semantic symbols to combat fading and noise, i.e., $\mathbf x_i(t)= E_{\nu}({\mathbf z}_i)$, where $\nu$ denotes the trainable parameters and $L_i(t)$ is the size of extracted semantic symbol vector from the $t$-th source video frame intended for user $i$. The channel bandwidth ratio (CBR) is $L_i (t)/3HW$ which measures the utilization of channel resources per symbol.
Therefore, the semantic information encoding process of the video frame can be summarized as
\begin{equation}
\mathbf{s}_i (t)\xrightarrow{E_{\varrho}(\cdot)}
\mathbf z_i (t)  \xrightarrow{E_{\nu}(\cdot)}
{\mathbf{x}}_i (t),
\end{equation}
where the semantic symbol vector ${\mathbf{x}}_i (t)$ is given by
\begin{equation}
\mathbf{x}_i(t) = {E_{\nu}(E_{\varrho}(\mathbf{s}_i (t)))}.
\end{equation}

\subsubsection{Wireless Channel Transmission}

We consider a Rayleigh block-fading channel model where each element of the channel matrix follows an independent and identically distributed (i.i.d.) complex Gaussian distribution $\mathcal {CN}(0,1)$.
The received signals $\mathbf {y}_i (t) \in \mathbb{C}^{N_r \times 1}$ at the $i$-th user at time slot $t$ can be written as 
\begin{equation}
\label{received signal-original}
\begin{aligned}
&\mathbf {y}_i(t) = \mathbf{H}_{i} (t) \sum_{j=1}^U \mathbf{V}_{j} (t)\mathbf{x}_{j}(t) + \mathbf{n}_i(t),
\end{aligned}
\end{equation}
where $\mathbf{H}_{i} (t) \in \mathbb{C}^{N_r \times N_t}$ represents the channel state information (CSI) between the BS and user ${i}$, $\mathbf{V}_{i} (t)$ denotes the transmit beamformer, while $\mathbf{n}_i(t)$ denotes the additive white Gaussian noise with distribution $\mathcal{CN}(0,\sigma_i^2\mathbf I)$. We assume the channel $\mathbf{H}_i (t)$ is block-fading and the channel coefficients remain constant for a block of consecutive symbols and change to an independent realization in the next block. We assume that the channel state information and source statistics are known to the transmitter and receiver. The beamformer $\mathbf V_i (t)$ satisfies the power constraint 
\begin{equation}
\sum_{i=1}^U\text{Tr}\left(\mathbf{V}_{i} (t)\mathbf{V}_{i}^H (t)\right) \leq P_\mathrm{max},
\end{equation}
where $P_\mathrm{max}$ denotes the power budget. The CSI is assumed to be known perfectly by both the transmitter and receiver.

\subsubsection{Semantic-aware Receiver}
At device user $i$, we consider a linear receive beamforming strategy so that the estimated signal is given by
\begin{equation}
\hat{\mathbf{x}}_i (t) = \mathbf{U}_{i}^H (t)\mathbf {y}_i(t), \forall i,
\end{equation}
where $\mathbf{U}_{i} (t) \in \mathbb{C}^{N_r \times N_r}$ is receive beamformer for user $i$ to detect $\mathbf x_i(t)$. Then the estimated signal passes through the channel decoder, ${E_{d}^{-1}(\cdot)}:\mathbb{C}^{L_i(t) \times 1}  \mapsto \mathbb{R}^{{N_H} \times 1}$ parameterized by $d$ to get the recover semantic information $\hat{\mathbf z}_i (t) \in \mathbb{R}^{{N_H} \times 1}$, and the semantic decoder, ${E_{e}^{-1}(\cdot)}: \mathbb{R}^{{N_H} \times 1} \mapsto  \mathbb{R}^{{N} \times 1}$ parameterized by $e$ to recover the original source frame $\hat{\mathbf s}_i (t) \in \mathbb{R}^{{N} \times 1}$. Semantic communication only transmits the relevant semantic information instead of the whole message to enhance transmission efficiency.
The above process is summarized as follows
\begin{equation}
\mathbf y_i (t)\xrightarrow{\mathbf{U}_i (t)}
\hat{\mathbf x}_i (t)\xrightarrow{E^{-1}_{d}(\cdot)}
\hat{\mathbf z}_i (t)\xrightarrow{E^{-1}_{e}(\cdot)}
\hat{\mathbf s}_i (t).
\end{equation}
Overall, the recovered semantic information $\hat{\mathbf z}_i(t)$ and the  reconstructed video frame $\hat{\mathbf s}_i(t)$ at the receiver side can be, respectively, expressed as
\begin{equation}
\label{eq: recovered1}
\begin{aligned}
\hat{\mathbf z}_i (t) &= E^{-1}_{d} (\hat{\mathbf x}_i (t))= E^{-1}_{d}(\mathbf{U}_{i}^H (t)\mathbf {y}_i(t))),\\
\hat{\mathbf s}_i (t) &=E^{-1}_{e} (\hat{\mathbf z}_i (t)) = {E^{-1}_{e}(E^{-1}_{d}(\mathbf{U}_{i}^H (t)\mathbf {y}_i(t)))}.
\end{aligned}
\end{equation}
Note that for any transmitted semantic symbol vector ${\mathbf{x}}_i (t)$, the   recovered semantic information $\hat{\mathbf z}_i(t)$ and the reconstructed video frame $\hat{\mathbf s}_i(t)$ are both determined by the transceiver beamformers $\{\mathbf V_i(t), \mathbf U_i(t)\}$, and the received signal $\mathbf y_i(t)$ involves all the transmit beamformers as defined in Eq. \eqref{received signal-original}.

Thanks for the semantic communication framework, the transmission delay $T_i(t)$ of the $t$-th source video frame at user $i$ is reduced to \begin{equation}
T_{i}(t)=\frac{{L}_i(t)}{B R_{i}(t)},
\end{equation}
where $B$ (Hz) is the bandwidth in the system; the transmission rate $R_{i}(t)$ (bps/Hz) is defined as Eq. \eqref{eq: rate}.
\begin{figure*}[hb]
\centering 
\hrulefill 
\begin{equation}
\label{eq: rate}
\begin{aligned}
&R_{i}(t)\triangleq \log\det\left(\mathbf{I}+ \mathbf{U}_{i}^H(t)\mathbf{H}_{i}(t)\mathbf{V}_{i}(t)\mathbf{V}_{i}^H(t)\mathbf{H}_{i}^H(t) \mathbf{U}_{i}(t)\left(\sum\limits_{j\neq i}\mathbf{U}_{i}^H(t)\mathbf{H}_{j}(t)\mathbf{V}_{j}(t)\mathbf{V}_{j}^H(t)\mathbf{H}_{j}^H(t)\mathbf{U}_{i}(t)+\sigma_{i}^2\mathbf{I}\right)^{-1}\right).
\end{aligned}
\end{equation}
\end{figure*}

\subsection{Age of Incorrect Semantics}
In this subsection, we analyze the design of the video frame transmission mechanism according to the age of incorrect semantics (AoIS) and drop index $i$ for convenience.

\begin{Definition}[AoIS]
Let $\mathbf z(t)$ {{denote}} the semantic information extracted by the semantic encoder and $\hat{\mathbf z}(t)$ {{represent}} the recovered semantic feature at a user, respectively, in communication round $t$. Then the age of incorrect semantic information at the user side is defined as \cite{AoII} \cite{Maatouk_2022_AoI}
\begin{equation}
\Delta_\text{AoIS} (\mathbf z(t), \hat{\mathbf z}(t), t) \triangleq f(t) g(\mathbf z(t), \hat{\mathbf z}(t)),
\end{equation}
where $g(\mathbf z(t), \hat{\mathbf z}(t)) \in [0,1]$ reﬂects the semantic mismatch between $\mathbf z(t)$ and $\hat{\mathbf z}(t)$. The non-decreasing function $f(t): [0, +\infty) \mapsto [0, +\infty)$ plays the role of penalizing the system processing time.
\end{Definition}
From the above definition, one can see that under the same semantic similarity condition, a larger value of $f(t)$ leads to a larger AoIS value.

The BS invariably keeps its status at the latest (current) version in the monitoring system. In order to evaluate the performance of semantic communications for video frame transmission, we adopt the semantic similarity, denoted as $\mathrm{Sim}(\mathbf{z}, \hat{\mathbf{z}})$, between the original feature $\mathbf{z}$ at the BS and the recovered feature $\hat{\mathbf{z}}$ at the receiver as the performance metric, i.e.,
\begin{equation}
\begin{aligned}
	&\mathrm{Sim}(\mathbf{z}, \hat{\mathbf{z}})=\frac{f_\lambda(\mathbf{z})f_\lambda(\hat{\mathbf{z}})^\mathrm{T}}{\|f_\lambda(\mathbf{z})\|\|f_\lambda(\hat{\mathbf{z}})\|},
\end{aligned}
\label{def_similarity1}
\end{equation}
which means $\mathbf{z}$ and $\hat{\mathbf{z}}$ are input into the downstream task execution module $f_\lambda(\cdot)$ \cite{wang_spiking}, and then calculate the cosine similarity of the inference results. From $\eqref{def_similarity1}$, we have $0 \leq \mathrm{Sim}(\mathbf{z}, \hat{\mathbf{z}}) \leq 1$, and the higher semantic similarity, the more task-related information is preserved in the reconstructed features.
To align with the system's optimization goal, we denote the optimization objective function $g(\mathbf{z}, \hat{\mathbf{z}})$ as
\begin{equation}
g(\mathbf{z}, \hat{\mathbf{z}}) = 1-\mathrm{Sim}(\mathbf{z}, \hat{\mathbf{z}}),
\end{equation}
where $0 \leq g(\mathbf{z}, \hat{\mathbf{z}}) \leq 1$. 
The existing works usually propose various schemes to ensure error-free data transmission, such as hybrid automatic repeat request (HARQ) to make $\mathbf z=\hat{\mathbf z}$. Once the data is successfully received and verified, an acknowledgment (ACK) signal is sent back to confirm the correct delivery. 
In this paper, we leverage semantic communication techniques to achieve robust performance even in the presence of data loss, i.e., $\mathbf z \neq \hat{\mathbf z}$. To circumvent the computational infeasibility of computing $g(\mathbf{z}, \hat{\mathbf{z}})$, we employ an approximation approach by taking the expectation over the noise distribution as a surrogate for the exact performance metric, i.e.,
\begin{equation}
\label{def_similarity}
\begin{aligned}
	&g(\mathbf{z}, \hat{\mathbf{z}})\approx 1- \mathbb E_{\mathbf n}\left[\frac{f_\lambda(\mathbf{z})f_\lambda(\hat{\mathbf{z}})^\mathrm{T}}{\|f_\lambda(\mathbf{z})\|\|f_\lambda(\hat{\mathbf{z}})\|}\right]\\
	&\approx 1- \frac{1}{M}\sum_{m=1}^M\frac{f_\lambda(\mathbf{z})f_\lambda(\hat{\mathbf{z}}^{(m)})^\mathrm{T}}{\|f_\lambda(\mathbf{z})\|\|f_\lambda(\hat{\mathbf{z}}^{(m)})\|},
\end{aligned}
\end{equation}
where the last approximation arises as a consequence of utilizing a Monte Carlo method for the estimation of the mathematical expectation. $M$ denotes the number of samples, and $\hat{\mathbf{z}}^{(m)}$ is the $m$-th reconstructed $\mathbf{z}$ at the BS as follows. At the BS, we regenerate the $m$-th received copy corresponding to $\mathbf{z}$ and the transmitted $\mathbf{x}$ by utilizing the known CSI and then adding a random noise after the channel with the CSI.

In this work, $f(t)$ is introduced to synchronize data transmission at intervals, compensating for the imprecision of semantic similarity estimates and ensuring timely updates to prevent excessive estimation errors due to prolonged unrefreshed data.
Generally, considering an exponential time-dissatisfaction function, $f(t)$ can be formulated as
\begin{equation}
f(t) = \exp(b(t - \epsilon(t))),
\end{equation}
where $b > 0$ is a positive constant, $\epsilon(t)$ denotes the last instant time where the user device receives the video frame. Thus, to reduce AoIS, it is essential to achieve both shorter update intervals and maintain higher semantic similarity accuracy. 

\subsection{Semantic Actuation Efficiency Policy Design}
For the time-slotted wireless MIMO communication system, the BS monitors a discrete random process and sends video frame sequences to users according to the status updates. The receiver sends an  ACK and negative ACK (NACK) for successful and failed transmissions, respectively. We assume the ACK/NACK transmission is instantaneous and error-free\footnote{This is a common assumption in the literature, as ACK/NACK packets are usually smaller than data packets and may be sent over a separate channel.}. As illustrated in Fig. \ref{system_2}, the BS sequentially parses each video frame in the video sequence $\mathbf S$. We denote the decision to sample the $t$-th video frame and transmit to user $i$ as $\alpha_i(t) \in \mathcal A$, $\mathcal A \in \{0,1\}$ as follows \cite{Pappas_goal}
\begin{equation}
\alpha_i(t)=\begin{cases}1,\text{ source sampled and state transmitted,}\\0,\text{ otherwise.}&\end{cases}
\end{equation}

If the transmission of frame $\mathbf s_i(t)$ is sampled and transmitted, i.e., $\alpha_i(t) =1$, then the AoIS at the subsequent time $t$ is represented by 
\begin{equation}
\begin{aligned}
\Delta_\text{AoIS} (\mathbf z_i(t), \hat{\mathbf z}_i(t), t)
&=\exp(b(t - \epsilon_i(t))) g(\mathbf{z}_i(t), \hat{\mathbf{z}}_i(t)).
\end{aligned}
\end{equation} 
Note that it is infeasible for the BS to compute the precise recovery information $\hat{\mathbf z}_i(t)$ at the $i$-th user due to the unavailability of accurate additive noise. Given the CSI and the noise distribution, we can estimate the approximate probability distribution of the recovered $\hat{\mathbf z}_i(t)$. Consequently,  the BS relies on the approximate AoIS for decision-making. This approximate AoIS is derived by taking the sample average as an estimate of the expectation in Eq. \eqref{def_similarity}.

Otherwise, the discrepancy between the frame at user $i$ and the frame at BS is small, there is no need to transmit this frame for channel bandwidth and power saving. Then, the receiver does not need an update, it uses its previous estimate as the current one, i.e., $\epsilon_i(t) = \epsilon_i(t-1) $ and $\hat{\mathbf z}_i(t) = \hat{\mathbf z}_i(t-1)$, and the evolution of AoIS is reduced to
\begin{equation}
\begin{aligned}
\Delta_\text{AoIS} (\mathbf z_i(t), \hat{\mathbf z}_i(t), t)
=&\exp(b(t - \epsilon_i(t-1)))\\
&\cdot g(\mathbf{z}_i(t), \hat{\mathbf{z}}_i(t-1)).
\end{aligned}
\end{equation}

Based on the analysis above, we can observe that to realize frame transmission, two conditions should be met: i) The AoIS between a frame at BS and that at a user does not exceed the predefined threshold; ii) It also should satisfy the cost constraints on system resources imposed by $\alpha_i(t)$.

\vspace{-0.5cm}
\subsection{Problem Formulation}
For every sampling and transmission action, we consider a cost of $c$. This cost can represent, for instance, the power consumption for both sampling and transmission procedures. For transmitting video frame $t$, the system incurs two costs: The first one is the average transmission decision cost $\bar c_i$ defined as 
\begin{equation}
\bar{c}_i\triangleq\lim_{T\to\infty}\frac1 {T}\sum_{t=0}^{T}\mathbb{E}_{\alpha_i(t)}\left\{\alpha_i(t)c_i(t)\right\},
\label{cost}
\end{equation}
where ${T}$ is the total number of time slots, and $c_i(t)$ is the cost at the $t$-th time slot. The cost $c_i(t)$ is modeled as a constant representing the average power consumption per transmission and sampling action, enabling us to constrain the transmission frequency to balance system performance and energy efficiency.
The second one is the actuation cost, which is mainly determined by transmission delay and the semantic similarity. The expected actuation cost over time is defined as 
\begin{equation}
\bar{C}\triangleq\lim_{{T}\to\infty}\frac1{T}\sum_{t=0}^{T}\sum\limits_{i=1}^{U}\mathbb{E}_{\alpha_i(t)}\left\{\Delta_\text{AoIS} (\mathbf z_i(t), \hat{\mathbf z}_i(t), t)\right\}.
\end{equation}

In this paper, we aim to minimize the time-averaged AoIS of all the users by jointly optimizing the semantic actuation indicator  $\alpha_i(t)$, transceiver beamformers $\mathbf{V}_{i}(t), \mathbf U_i(t)$, and the semantic symbol length $L_i(t)$, subject to the system transmission decision cost, the transmit power constraints, the allocated number of semantic symbols per frame, and the maximum time delay tolerance. The AoIS-aware resource allocation problem is formulated as
\begin{subequations}
\label{p: initial}
\begin{align}
\min\limits_{\alpha_i(t), L_i(t), \mathbf{V}_{i}(t), \mathbf U_i(t)}
&\bar{C} \\\mathrm{s.t.}\,\,\,\,\,\,\,\,\,
&\bar{c}_i \leq c_{\mathrm {max}},\forall i \label{a}\\
&\sum\limits_{i=1}^{U}\text{Tr}\left(\mathbf{V}_{i}(t)\mathbf{V}_{i}^H(t)\right)  \leq P_\mathrm{max}, \label{d}\\
& L_i(t) \in \mathcal L,\forall i, \label{e}\\
& T_{i}(t) \leq T_\mathrm{max},\forall i, \label{g}
\end{align}
\end{subequations}
where $c_{\mathrm {max}}$ in \eqref{a} is a maximum actuation cost. \eqref{d} is the transmission power constraint at the BS. \eqref{e} $\mathcal L$ specifies all possible semantic symbol length per frame, and \eqref{g} restricts the maximum time delay tolerance of users by $T_\mathrm{max}$.

\emph{Remark:} The computational cost is usually proportional to the size of the input resolution, i.e., $N = 3HW$, such as CNN  encoder and decoder \cite{He}. The required number of CPU cycles per pixel using the encoder is determined by the architecture of the encoder and decoder networks. Once given the allocated computing resource to user $i$, the computational latency of the encoding and decoding processes is a fixed value. Therefore, we mainly consider the transmission delay as an important latency metric of wireless communication.

\section{Semantic-aware Resource Allocation}
\label{sec:perfect}
In this section, we first minimize the average total cost of semantic actuation error under average resource constraints. Then we resort to the AO algorithm to solve Problem \eqref{p: initial}. The AO algorithm decomposes the optimization variables into multiple blocks and optimizes different blocks in each iteration, with the other blocks fixed.

\subsection{Lyapunov Optimization}
Our objective is to devise a low-complexity algorithm to ensure that the average cost constraints are met while approximating the optimal solution. By leveraging Lyapunov optimization techniques, we introduce a real-time algorithm termed Drift-Plus-Penalty (DPP).
In order to satisfy the average cost constraints, we map the average cost constraint in \eqref{a} into a virtual queue \cite{DPP}. This allows us to equivalently transform the cost constraint satisfaction problem into a queue stability problem as detailed below.

We first define a virtual queue associated with the constraint in \eqref{a}, where a stable virtual queue implies a constraint-satisfying result. Let ${Q_i(t)}$ be the virtual queue at every video frame $t$ with dynamic update equation \cite[Eq. (11)] {Pappas_goal}
\begin{equation}
Q_i(t+1)=\mathrm{max}[Q_i(t)-c_\mathrm{max},0]+\alpha_i(t)c_i(t),
\label{Q_{t+1}}
\end{equation}
where the non-negative real valued random variable $Q_i(0)=0$. Based on the fundamental Lyapunov drift theorem \cite{Lyapunov}, process $\{Q_i(t)\}$ can be viewed as a virtual queue with arrivals $\alpha_i(t)$ and service rate $c_i(t)$. $\alpha_i(t)c_i(t)$ represents the amount of new work that arrives on slot $t$, and it is assumed to be non-negative.
To that end, we provide the following proposition.

\begin{Proposition}
\label{Pro}
If the virtual queue $Q_i(t)$ is mean rate stable, i.e.,
\begin{equation}
\lim_{t\to\infty}\frac{Q_i(t)} t=0,
\end{equation}
then the constraint $\bar{c}_i \leq c_{\mathrm {max}}$ in \eqref{a} is satisfied.
\end{Proposition}

\begin{proof}
See Appendix \ref{Appendix_Pro}.
\end{proof}

Based on Proposition \ref{Pro}, we aim to stabilize the virtual queue $Q_i(t)$. Denote $\mathbf Q(t) = (Q_1(t), Q_2(t), \cdots, Q_U(t))$ as the vector of all device users' virtual queues to measure the queue backlogs.
Let ${\boldsymbol{\alpha}(t)} = ({\alpha_1}(t), {\alpha_2}(t), \cdots, {\alpha_U}(t))$ represent all users' semantic actuation decision at time slot $t$.
We adopt the quadratic Lyapunov function $\Gamma(\mathbf Q(t)) = \frac{1}{2}\sum_{i=1}^U (Q_i(t))^2$ as a scalar measure of queue congestion. The constant coefficient $1/2$ can facilitate the deduction and algorithm design in the following part. A small value of $\Gamma(\mathbf Q(t))$ indicates less congestion and better satisfaction. Then, the one-slot Lyapunov drift is introduced to describe the small queue variation between two successive video frame slots, defined as
\begin{equation}
\Delta(\mathbf Q(t))\triangleq\mathbb{E}_{\boldsymbol{\alpha}(t)}\big\{\Gamma(\mathbf Q(t+1))-\Gamma(\mathbf Q(t))\big|\mathbf Q(t)\big\}.
\end{equation}
To stabilize the virtual AoIS queue $\mathbf Q(t)$, we aim to minimize the increment of the queue size, i.e., the Lyapunov drift $\Delta(\mathbf Q(t))$.
Given the inequality $(\mathrm{max}[W - b, 0] + a)^2 \leq W^2 + a^2 + b^2 + 2W (a-b)$ for non-negative $a$ and $b$, we have $\Delta(\mathbf Q(t))$ for a general policy satisfies
\begin{equation}
\begin{aligned}
\Delta & (\mathbf Q(t))
= \mathbb{E}_{\boldsymbol{\alpha}(t)}\big\{\Gamma(\mathbf Q(t+1))-\Gamma(\mathbf Q(t))\big|\mathbf Q(t)\big\} \\
&= \frac{1}{2}\sum_{i=1}^U \mathbb{E}_{{\alpha_i}(t)} \bigg\{ \bigg(\mathrm{max}[Q_i(t)-c_\mathrm{max},0]+\alpha_i(t)c_i(t)\bigg)^2  \\
& \qquad \qquad \qquad - (Q_i(t))^2 \big| Q_i(t) \bigg\}
\\
&\leq B+\sum_{i=1}^U\mathbb{E}_{{\alpha_i}(t)}\left\{Q_i(t)(\alpha_i(t)c_i(t)-c_\mathrm{max})\big|Q_i(t)\right\},
\end{aligned}
\end{equation}
where the constant $B = \sum_{i=1}^U \mathbb{E}_{{\alpha_i}(t)} \{((\alpha_i(t)c_i(t))^2 + c_\mathrm{max}^2)/2\}$.
We apply the DPP algorithm to minimize the average expected cost to keep information fresh while stabilizing the virtual queues $\mathbf Q(t)$. Specifically, this problem is reduced to minimizing the upper bound of the following expression
\begin{equation}
\begin{aligned}
&\Delta(\mathbf Q(t)) + \omega \sum\limits_{i=1}^{U}\mathbb{E}_{{\alpha_i}(t)}\left\{\Delta_\text{AoIS} (\mathbf z_i(t), \hat{\mathbf z}_i(t), t)\right\}\\
&\leq B + \sum_{i=1}^U\mathbb{E}_{{\alpha_i}(t)}\left\{Q_i(t)(\alpha_i(t)c_i(t)-c_\mathrm{max})\big|Q_i(t)\right\}\\
& + \omega  \sum\limits_{i=1}^{U}\mathbb{E}_{\mathbf n_i(t),{{\alpha_i}(t)}}\left\{\Delta_\text{AoIS} (\mathbf z_i(t), \hat{\mathbf z}_i(t), t)\right\} \triangleq \tilde C,
\end{aligned}
\label{upper_bound}
\end{equation}
where $\omega > 0$ is a weight.
To this end, we aim to minimize the right-hand side of the inequality \eqref{upper_bound} and problem \eqref{p: initial} is transformed into a per-slot optimization problem based on the Lyapunov optimization framework. Based on the preceding analysis, the optimization problem can be reformulated for each video frame and streamlined as 
\begin{subequations}
\label{p:v1}
\begin{align}
\min\limits_{\{\alpha_i(t),  L_i(t), \mathbf{V}_{i}(t), \mathbf U_i(t)\}}
& \tilde{C} \\\mathrm{s.t.}\,\,\,\,\,\,\,\,\,
&\sum\limits_{i=1}^{U}\text{Tr}\left(\mathbf{V}_{i}(t)\mathbf{V}_{i}^H(t)\right)  \leq P_\mathrm{max}, \label{d1}\\
& L_i(t) \in \mathcal{L},\forall i, \label{e1}\\
& T_{i}(t) \leq T_\mathrm{max},\forall i. \label{g1}
\end{align}
\end{subequations}

\emph{Remark:} The DPP algorithm has low complexity and is scalable for large systems, while it is suboptimal due to the ignorance of both long-term performance and the requirement of prior knowledge of channel/source statistics to compute expected action costs.

\subsection{Semantic Actuation Policy Design}
According to the principle of AO, we first investigate the optimization of transmitting decision indicator $\alpha_i(t)$ for fixed transmit beamformer $\mathbf V_i(t)$, receive beamformer $\mathbf U_i(t)$, and the semantic symbol length $L_i(t)$. For every time slot $t$, the transmitting decision indicator design
problem is accordingly reduced to
\begin{equation}
\label{p:a}
\begin{aligned}
&\min\limits_{\{{\alpha_i}(t)\}}
\,\,\,\sum_{i=1}^U\mathbb{E}_{{\alpha_i}(t)}\left\{Q_i(t)(\alpha_i(t)c_i(t)-c_\mathrm{max})\big|Q_i(t)\right\} \\
& + \omega \sum\limits_{i=1}^{U}\mathbb{E}_{{\alpha_i}(t)}\left\{\exp(b(t - \epsilon_i(t))) g(\mathbf z_i(t), \hat{\mathbf z}_i(t))\right\}.
\end{aligned}
\end{equation}
We adopt an exhaustive search method to solve the above problem for the optimal $\alpha_i(t)$. Since each $\alpha_i(t)$ is binary, there are a total of $2^U$ candidate combinations for $\boldsymbol{\alpha}(t) = (\alpha_1(t), \ldots, \alpha_U(t))$. Specifically, for each time slot $t$, the exhaustive search method systematically evaluates all possible combinations of $\alpha_i(t)$ values, computes the corresponding objective function value for each combination, and selects the set of $\alpha_i(t)$ values that minimizes the objective function.
This method guarantees a thorough investigation of every viable solution, ensuring a satisfactory performance while preserving a manageable computational load.

\subsection{Transceiver Beamformer Design}
According to the principle of AO, we optimize the beamformer $\mathbf{V}_{i}(t)$ for fixed $\alpha_i(t)$, ${L}_i(t)$, and $\mathbf U_i(t)$. The beamformer design problem is accordingly reduced to 
\begin{subequations}
\label{p: transmit}
\begin{align}
\min\limits_{\{\mathbf{V}_{i}(t)\}}\quad
& J(\{\mathbf{V}_{i}(t)\})
\label{V1_a}\\
\mathrm{s.t.}\,\,\,\,
& \sum\limits_{i=1}^{U}\text{Tr}\left(\mathbf{V}_{i}(t)\mathbf{V}_{i}^H(t)\right)  \leq P_\mathrm{max},  \label{V1_b}\\
&R_{i}(t) \geq \frac{{L}_i(t)}{BT_\mathrm{max}}, \forall i.  \label{V1_c}
\end{align}
\label{V}
\end{subequations}
where $J(\{\mathbf{V}_{i}(t)\})\triangleq \sum\limits_{i=1}^{U}\exp(b(t - \epsilon_i(t))) g(\mathbf z_i(t), \hat{\mathbf z}_i(t))$.

We first propose to leverage the SCA \cite{SCA} technique to update the transceiver beamforming matrix. The key to the success of SCA lies in constructing a surrogate function satisfying the gradient consistency for the objective function \eqref{V1_a} and the upper bound property for the constraints \eqref{V1_c}. Now, we introduce a lemma that facilitates the implementation of SCA.

\begin{lemma}
\label{lemma: log}
The logarithmic form $\log\det(\mathbf{I}+\mathbf{X}^{H}\mathbf{Y}^{-1}\mathbf{X})$ can be lower bounded as \cite[Prop. 7]{prop7}:

\begin{equation}
\begin{aligned}
\log & \det(\mathbf{I}+\mathbf{X}^{H}\mathbf{Y}^{-1}\mathbf{X}) \geq -\operatorname{tr}\big((\mathbf{Y}_0+\mathbf{X}_0\mathbf{X}_0^{H})^{-1}{\mathbf{X}_0} \\ & \times(\mathbf{I}+\mathbf{X}_0^{H}\mathbf{Y}_0^{-1}\mathbf{X}_0)\mathbf{X}_0^H(\mathbf{Y}_0+\mathbf{X}_0\mathbf{X}_0^H)^{-1} (\mathbf{Y} + \mathbf{X} \mathbf{X}^H)\big) \\
& + 2\Re \left(\operatorname{tr}((\mathbf{I}+\mathbf{X}_0^H\mathbf{Y}_0^{-1}\mathbf{X}_0)\mathbf{X}_0^H(\mathbf{Y}_0+\mathbf{X}_0\mathbf{X}_0^H)^{-1} \mathbf{X})\right) \\
&+ \log\det(\mathbf{I}+\mathbf{X}_0^H\mathbf{Y}_0^{-1}{\mathbf{X}_0}) -\mathrm{tr}(\mathbf{X}_0^H\mathbf{Y}_0^{-1}{\mathbf{X}_0}),    
\end{aligned}
\end{equation}
and the equality is achieved at $\mathbf{X}=\mathbf{X}_0, \mathbf{Y}= \mathbf{Y}_0$.
\end{lemma}

\begin{algorithm}[t]
\caption{AoIS-aware Resource Allocation Design.}
\label{Alg: SCA}
{\bf Initialization:} Channel matrices $\mathbf H_i(t), \forall t, \forall i \in U$; the power budget $P_\mathrm{max}$; the maximum transmission delay tolerance $T_\mathrm{max}$.
\begin{algorithmic}[1]
\FOR {$t=1,2, \cdots$}
\STATE \textbf{Step 1:} DPP policy for problem transformation.\\
\text{~~~} {1.1:} Set weight factor $\omega$, and initialize the virtual queue 
\text{~~~~~~~~}  $Q_0(t) = 0$.\\
\text{~~~} {1.2:} At the start of the transmission of the $t$-th frame, \text{~~~~~~~~} observe $Q_i(t)$.\\
\text{~~~} {1.3:} Calculate $\bar c_i$ using equation \eqref{cost}.\\
\text{~~~} {1.4:} Make a selection of action $\alpha_i^*(t)$ based on \eqref{p:a}.\\
\text{~~~} {1.5:} Apply action $\alpha_i^*(t)$ and update $Q(t)$ according \text{~~~~~~~~}  to \eqref{Q_{t+1}}.

\STATE \textbf{Step 2:} AoIS-aware SCA algorithm to design \text{~~~~~~~~~~~~~} transceiver beamforming.\\
\text{~~~}  {2.1:} Initialize $\mathbf{V}_{i}^{(0)}(t)$, $\mathbf{U}_{i}^{(0)}(t)$, $\gamma_{v}^{(0)}(t)$, $\gamma_{u}^{(0)}(t)$, and \text{~~~~~~~~} set $n=0$.\\

\text{~~~} {2.2:} Construct surrogate functions \\
\text{~~~~~~~~}  $\tilde{J}(\{\mathbf{V}_{i}(t), \mathbf{V}_{i}^{(n)}(t)\})$ and $\tilde{J}(\{\mathbf{U}_{i}(t), \mathbf{U}_{i}^{(n)}(t)\})$ \\
\text{~~~~~~~~}  using Lemma \ref{lemma: log}.\\

\text{~~~}  {2.3:} Solve the convex QP problem in \eqref{Problem: V} and \eqref{Problem: U}.\\

\text{~~~}  {2.4:}  Update $\mathbf{V}_{i}^{(n+1)}(t)$ and $\mathbf{U}_{i}^{(n+1)}(t)$ according to Eqs. \eqref{update: V} and \eqref{update: U}, and go to \textbf{Step 2}, 2.2. This process is repeated until convergence is achieved.

\STATE \textbf{Step 3:} Optimize Problem \eqref{p: L} with the exhaustive search algorithm.
\ENDFOR
\end{algorithmic}
\end{algorithm}

Based on Lemma \ref{lemma: log}, at iteration $n$, a majorizing function for the constraint is constructed at $\mathbf V_{i}(t) = \mathbf V_{i}^{(n)}(t)$ as follows:
\begin{equation}
\begin{aligned}
R_{i}(t) & \geq   \Re \left(\operatorname{tr}(\mathbf F_2(t) \mathbf{U}_i^H(t)\mathbf{H}_{i}(t)\mathbf{V}_{i}(t))\right)   \\
&\quad - \operatorname{tr}\Bigg(\mathbf F_1(t) \bigg(\sum_{m=1}^{U}\mathbf{U}_{i}^H(t)\mathbf{H}_{i}(t)\mathbf{V}_{m}(t) \\
& \qquad  \cdot \mathbf{V}_{m}^H(t)\mathbf{H}_{i}^H(t)\mathbf{U}_{i}(t)  +\sigma_{i}^2\mathbf{I}\bigg)\Bigg) + \text{const}   \\
& \triangleq \tilde{R}_{i} (\{\mathbf{V}_{i}(t), \mathbf{V}_{i}^{(n)}(t)\}),
\end{aligned}
\end{equation}
where
\begin{subequations}
\begin{align}
\mathbf F_1(t) &\triangleq \mathbf F_3^{-1}(t){\mathbf{X}_0(t)} \mathbf F_4(t)\mathbf{X}_0^H(t) \mathbf F_3^{-1}(t),\\
\mathbf F_2(t) &\triangleq 2\mathbf F_4(t)\mathbf{X}_0^H(t)\mathbf F_3^{-1}(t), \\
\mathbf F_3(t) &\triangleq \mathbf{Y}_0(t) +\mathbf{X}_0(t)\mathbf{X}_0^{H}(t),\\
\mathbf F_4(t) &\triangleq \mathbf{I}+\mathbf{X}_0^{H}(t)\mathbf{Y}_0^{-1}(t)\mathbf{X}_0(t),\\
\mathbf{X}_0(t) &\triangleq \mathbf{U}_i^H(t)\mathbf{H}_{i}(t)\mathbf{V}_{i}^{(n)}(t),\\
\mathbf{Y}_0(t) \triangleq &\Bigg(\sum\limits_{m\neq i}\mathbf{U}_{i}^H(t)\mathbf{H}_{i}(t)\mathbf{V}_{m}^{(n)}(t)\\& \qquad \qquad \cdot \mathbf{V}_{m}^{(n)H}(t)\mathbf{H}_{i}^H(t)\mathbf{U}_{i}(t)+\sigma_{i}^2\mathbf{I}\Bigg)^{-1}.
\end{align}
\end{subequations}

Therefore, the transmit beamforming matrix set $\{\mathbf V_{i}(t)\}$ is updated as 
\begin{subequations}
\label{Problem: V}
\begin{align}
\{\hat{\mathbf{V}}_{i}^{(n)}(t)\}
&\triangleq\arg\min\limits_{\{\mathbf{V}_{i}(t)\}}\quad
\tilde{J}(\{\mathbf{V}_{i}(t), \mathbf{V}_{i}^{(n)}(t)\})\\
\mathrm{s.t.}\,\,\,\,
& \sum\limits_{i=1}^{U}   \text{Tr}\left(\mathbf{V}_{i}(t)\mathbf{V}_{i}^H(t)\right)  \leq P_\mathrm{max},\\
& \tilde{R}_{i} (\{\mathbf{V}_{i}(t), \mathbf{V}_{i}^{(n)}(t)\}) \geq \frac{{L}_i(t)}{BT_\mathrm{max}},
\end{align}
\end{subequations}
where the surrogate function $\tilde{J}(\cdot)$ for the objective function is constructed at $\mathbf V_{i}(t) = \mathbf V_{i}^{(n)}(t)$ as follows:
\begin{equation}
\begin{aligned}
&\tilde{J}(\{\mathbf{V}_{i}(t), \mathbf V_{i}^{(n)}(t)\}) \\
&\triangleq  \sum\limits_{i=1}^{U} \exp(b(t - \epsilon_i(t))) \nabla g(\mathbf z_i(t), \hat{\mathbf z}_i^{(n)}(t)) \mathbf{V}_{i}(t),
\end{aligned}
\end{equation}
and from Eq.~\eqref{eq: recovered1}, $\nabla g(\mathbf z_i(t), \hat{\mathbf z}_i^{(n)}(t))$ denotes the gradient of $g(\mathbf z_i(t), \hat{\mathbf z}_i^{(n)}(t))$ in terms of $\mathbf{V}_{i}(t)$ evaluated at $\mathbf{V}_{i}(t) = \mathbf{V}_{i}^{(n)}(t)$, where $g(\mathbf z_i(t), \hat{\mathbf z}_i^{(n)}(t))$ is treated as a function of $\mathbf V_i(t)$. This gradient is obtained by using Python's automatic differentiation tools.
Given $\tilde{J}(\{\mathbf{V}_{i}(t), \mathbf V_{i}^{(n)}(t)\})$, $\mathbf{V}_{i}^{(n+1)}(t)$ is updated according to
\begin{equation}
\label{update: V}
\mathbf{V}_{i}^{(n+1)}(t) = \mathbf{V}_{i}^{(n)}(t) + \gamma_{v}^{(n)}(t) (\hat{\mathbf{V}}_{i}^{(n)}(t) - \mathbf{V}_{i}^{(n)}(t)),
\end{equation}
where $\gamma_{v}^{(n)}(t)$ is the iterative step size satisfying
\begin{equation}
\begin{aligned}
&\lim_{n\to\infty}\gamma_{v}^{(n)}(t)=0,\\
&\sum\limits_{n=0}^{\infty} \gamma_{v}^{(n)}(t) = \infty,\\
&\sum\limits_{n=0}^{\infty} (\gamma_{v}^{(n)}(t))^2 < \infty
\end{aligned}
\label{gamma_con}
\end{equation}
We can observe that the problem is a convex quadratic programming (QP) problem, which can be effectively solved by CVXPY.

For the optimization of the receive beamformer matrix set $\{\mathbf{U}_i(t)\}$, we adopt the SCA method and the optimization process is similar to $\{\mathbf{V}_{i}(t)\}$, which can be described as follows
\begin{subequations}
\label{Problem: U}
\begin{align}
\{\hat{\mathbf{U}}_{i}^{(n)}(t)\}
&\triangleq\arg\min\limits_{\{\mathbf{U}_{i}(t)\}}\quad
\tilde{J}(\{\mathbf{U}_{i}(t), \mathbf{U}_{i}^{(n)}(t)\})\\
\mathrm{s.t.}\,\,\,\,
& \tilde R_{i} (\{\mathbf{U}_{i}(t), \mathbf{U}_{i}^{(n)}(t)\}) \geq \frac{{L}_i(t)}{BT_\mathrm{max}},
\end{align}
\end{subequations}
where 
\begin{equation}
\begin{aligned}
&\tilde{J}(\{\mathbf{U}_{i}(t), \mathbf U_{i}^{(n)}(t)\}) \\&= \sum\limits_{i=1}^{U}\exp(b(t - \epsilon_i(t)))
\nabla \tilde g(\mathbf z_i(t),
\hat{\mathbf z}_i^{(n)}(t))\mathbf{U}_{i}(t),
\end{aligned}
\end{equation}
and, similarly, $\nabla \tilde g(\mathbf z_i(t), \hat{\mathbf z}_i^{(n)}(t))$ denotes the gradient of $g(\mathbf z_i(t), \hat{\mathbf z}_i^{(n)}(t))$ in terms of $\mathbf{U}_{i}(t)$ evaluated at $\mathbf{U}_{i}(t) = \mathbf{U}_{i}^{(n)}(t)$, where $g(\mathbf z_i(t), \hat{\mathbf z}_i^{(n)}(t))$ is treated as a function of $\mathbf U_i(t)$.
Given $\tilde{J}(\{\mathbf{U}_{i}(t), \mathbf U_{i}^{(n)}(t)\})$, $\mathbf{U}_{i}^{(n)}(t)$ is updated according to 
\begin{equation}
\label{update: U}
\mathbf{U}_{i}^{(n+1)}(t) = \mathbf{U}_{i}^{(n)}(t) + \gamma_{u}^{(n)}(t) (\hat{\mathbf{U}}_{i}^{(n)}(t) - \mathbf{U}_{i}^{(n)}(t)),
\end{equation}
where $\gamma_{u}^{(n)}(t)$ is the iterative step size satisfying the similar condition as \eqref{gamma_con}.

\subsection{Transmitted Semantic Symbols Design}
With the given user association, beamforming design, and detection scheme, Problem \eqref{p: initial} can be simplified as
\begin{subequations}
\label{p: L}
\begin{align}
\min\limits_{\{L_i(t)\}}\quad
&\sum\limits_{i=1}^{U}\exp(b(t - \epsilon_i(t))) g(\mathbf z_i(t),
\hat{\mathbf z}_i(t))\\
& L_i(t) \in \mathcal L,\forall i,\\
&R_{i}(t) \geq \frac{{L}_i(t)}{BT_\mathrm{max}},\forall i \label{p: L-c}.
\end{align}
\end{subequations}
We adopt an exhaustive search algorithm to solve the above problem. 
A convolution layer is adopted at the end of the JSCC encoder to adjust the semantic symbol length by varying the number of convolution filters.
Note that the complexity is acceptable due to the finite set of $\{L_i(t)\}$. This method guarantees a thorough investigation of every viable solution, ensuring satisfactory performance while preserving a manageable computational load.
For practical deployment, we first train a set of statistically optimal semantic encoder-decoder pair, and then perform real-time fine-tuning to realize problem optimization based on the CSI.
The above optimization procedure is summarized in Algorithm \ref{Alg: SCA}.

\section{Low-Complexity Design for a Special MISO System} 
\label{sec:LCZF}
In this section, we examine a specific scenario in which each user is equipped with only a single antenna, i.e., $N_r =1$. Under this condition, the CSI and the beamforming matrix are simplified to vector forms. Then, a low-complexity ZF algorithm is proposed to optimize the transceiver beamformer.

\subsection{Low-Complexity Algorithm Design}
Let $\mathbf h_i \in \mathbb C^{N_t \times 1}$ denote the CSI between the BS and user $i$ and $\mathbf v_i \in \mathbb C^{N_t \times 1}$ represent the transmit beamformer for BS to convey the semantic signal. Furthermore, the achievable transmission rate at user $i$ is reduced to
\begin{equation}
\label{eq: SNR}
r_{i}(t)\triangleq \log_2\left(1 + \frac{|\mathbf h_{i}^H(t) \mathbf v_{i}(t)|^2}{\sum\limits_{j\neq i} |\mathbf h_{i}^H(t) \mathbf v_{j}(t)|^2 + \sigma_i^2}\right).
\end{equation}
Hence, the transceiver beamforming problem is reduced to
\begin{subequations}
\label{p: MISO transmit}
\begin{align}
\min\limits_{\{\mathbf{v}_{i}(t), u_{i}(t)\}}\quad
& \Xi(\{\mathbf{v}_{i}(t), \mathbf{u}_{i}(t)\})\triangleq \notag
\\&\sum\limits_{i=1}^{U}\exp(b(t - \epsilon_i(t))) g(\mathbf z_i(t), \hat{\mathbf z}_i(t))
\label{V2_a}\\
\mathrm{s.t.}\,\,\,\,
& \sum\limits_{i=1}^{U}\|\mathbf v_{i}(t)\|^2\leq P_\mathrm{max},  \label{V2_b}\\
&r_{i}(t) \geq \frac{{L}_i(t)}{BT_\mathrm{max}},  \label{V2_c}
\end{align}
\end{subequations}
where $\mathbf{u}(t) \triangleq  \{u_1(t), u_2(t), \cdots, u_U(t)\}$ represents the receive beamformer.

\begin{algorithm}[t]
	\caption{Low-complexity ZF beamforming algorithm.}
	\label{Alg: ZF}
	{\bf Initialization:} Channel matrices $\mathbf h_i(t), \forall t, \forall i \in U$; the transmit power constraint $P_\mathrm{max}$; the maximum transmission delay tolerance $T_\mathrm{max}$.
	\begin{algorithmic}[1]
		\FOR {$t=1,2, \cdots$}
		\STATE \textbf{Step 1:} Utilize DPP algorithm for problem transformation and derive the reformulated problem \eqref{p:v1}.\\
		\STATE \textbf{Step 2:} The proposed AoIS-aware low-complexity ZF beamforming algorithm for solving Problem \eqref{p: transmit}.\\
		
		\text{~~~} {2.1:} Initialize optimization variable vectors $\mathbf{q}$ and $\mathbf{u}$\\
		\text{~~~} {2.2:} Transform Problem \eqref{p: transmit} into Problem \eqref{p  ZF beamforming 2}. \\
		\text{~~~} {2.3:} Optimize Problem \eqref{p  ZF beamforming 2} with SGD or Adam optimizer.
		
		\STATE \textbf{Step 3:} Optimize Problem \eqref{p:a} and get $\alpha_i^*(t)$.\\
		\STATE \textbf{Step 4:} Adopt the exhaustive search algorithm to solve Problem \eqref{p: L} to obtain $L_i^*(t)$.
		\ENDFOR
	\end{algorithmic}
\end{algorithm}

\vspace{-0.2cm}
To make the above problem more tractable, we assume that all involved channels are perfectly known at BS that employs ZF beamforming for transmission, which is known to be optimal in the high SINR regime.  Define $\mathbf V(t)\triangleq\left[\mathbf v_1(t),\mathbf v_2(t),\cdots,\mathbf v_U(t)\right]$. Then, the perfect interference suppression is achieved by setting the ZF beamforming matrix to 
\begin{equation}
\begin{aligned}
\mathbf V(t) = \mathbf H(t) (\mathbf H^H(t)\mathbf H(t))^{-1} \sqrt{\mathbf P(t)},
\end{aligned}
\end{equation} 
where $\mathbf H(t) \triangleq \left(\mathbf h_1(t), \mathbf h_2(t), \cdots, \mathbf h_{U}(t)\right)$. $\mathbf P(t) \triangleq \operatorname{diag}(\mathbf p(t))$, $\mathbf p(t) \triangleq \left(p_1(t), p_2(t), \cdots, p_{U}(t)\right)^T$, $p_i(t)$ denotes the transmission power allocated to the $i$-th user at the $t$-th frame. From the construction of $\mathbf V(t)$, we know $\mathbf H^H(t)\mathbf V(t) = \sqrt{\mathbf P(t)}$ is a diagonal matrix, which means
\begin{equation}
	\vspace{-0.1cm}
\mathbf h_i^H(t) \mathbf v_j(t)=\left\{
\begin{array}{cc}
\sqrt{p_i(t)}, &  i=j,\\
0,     & i\neq j.
\end{array}
\right.
\vspace{-0.1cm}
\end{equation}
Substituting this $\mathbf V(t)$ into the received signal, it can be expressed as follows:
\begin{equation}
\begin{aligned}
&\mathbf y_i(t) = \mathbf h_i^H(t)\sum_{i=1}^U \mathbf v_i(t)\mathbf s_i(t) + n_i(t)\\
&= \sqrt{p_i(t)} \mathbf x_i(t) + n_i(t).
\end{aligned}
\end{equation}
Then, the achievable rate of the user $i$ can be computed as
\begin{equation}
r_i(t) = \log_2 \left(1+ \frac{p_i(t)}{\sigma_i^2}\right),
\end{equation}
Therefore, the constraint \eqref{V2_c} is equivalent to
\begin{equation}
p_i(t) \geq 2^{L_{i}(t)/(BT_\mathrm{max})}-1.
\end{equation}
Besides, the transmit power is given by
\begin{equation}
\begin{aligned}
&P =  \sum\limits_{i=1}^{U}   \|\mathbf v_{i}(t)\|^2 \\
&=\operatorname{Tr}(\mathbf V(t) \mathbf V(t)^H)\\
&=\operatorname{Tr}\left(\mathbf H(t) (\mathbf H^H(t)\mathbf H(t))^{-1} \mathbf P(t) (\mathbf H^H(t)\mathbf H(t))^{-1} \mathbf H^H(t)\right)\\
&=\operatorname{Tr}\left((\mathbf H(t)\mathbf H^H(t))^{-1} \mathbf P(t)\right).
\end{aligned}
\end{equation}
Let $\mathbf h(t) \triangleq \operatorname{diag}\left((\mathbf H(t)\mathbf H^H(t))^{-1}\right)$. The constraint \eqref{V2_b} is equivalent to
\begin{equation}
\mathbf h(t)^T\mathbf p(t) \leq P_\mathrm{max}.
\end{equation}

Finally, the transceiver beamformer design problem \eqref{p: MISO transmit} is reduced to
\begin{subequations}
\label{ZF beamforming 1}
\begin{align}
\underset{\{\mathbf p(t), \mathbf u(t)\}}{\text{min}}  \quad & \Xi^\prime(\{\mathbf p(t), \mathbf u(t)\})\notag\\
&\triangleq \Xi(\mathbf H(t)(\mathbf H^H(t)\mathbf H(t))^{-1} \sqrt{\mathbf P(t)}, \mathbf u(t))\label{p ZF beamforming 1-a}\\
\text { s.t. } 
&\mathbf h^T(t)\mathbf p(t) \leq P_\mathrm{max},\label{p ZF beamforming 1-b}\\
& p_i(t) \geq 2^{L_{i}(t)/(BT_\mathrm{max})}-1, \forall i. \label{p ZF beamforming 1-c}
\end{align}
\end{subequations}


To solve this problem efficiently, we first derive the structure of $\mathbf p(t)$ in the following proposition.
\begin{Proposition}
\label{prop: ZF}
The constraints \eqref{p ZF beamforming 1-b}-\eqref{p ZF beamforming 1-c} are equivalent to
\begin{equation}
\label{eq: constraints equivalence}
\begin{aligned}
\mathbf p(t) &= \mathbf p_0(t) + \frac{\mathbf q(t)\odot \mathbf q(t)}{\mathbf h^T(t) (\mathbf q(t) \odot \mathbf q(t))}(P_\mathrm{max} - \mathbf h^T(t) \mathbf p_0(t)),
\end{aligned}
\end{equation}
where $\mathbf p_0(t) \triangleq (2^{L_{1}(t)/(BT_\mathrm{max})}-1,\cdots, 2^{L_{i}(t)/(BT_\mathrm{max})}-1, \cdots, 2^{L_{U}(t)/(BT_\mathrm{max})}-1)$. $\mathbf p(t)$ is determined by the introduced parameter $\mathbf q(t)$.
\end{Proposition}
\begin{proof}
Please see Appendix \ref{appendix prop: ZF}.
\end{proof}

Based on Prop.~\ref{prop: ZF}, the optimization problem \eqref{ZF beamforming 1} w.r.t. the variables $\{\mathbf p(t),\mathbf u(t)\}$ can be equivalently reformulated as an optimization problem w.r.t. the variables $\{\mathbf q(t),\mathbf u(t)\}$. Then the transceiver beamformer design problem is transformed into an unconstrained optimization problem as follows
\begin{equation}
\label{p ZF beamforming 2}
\begin{aligned}
\underset{\{\mathbf q(t), \mathbf u(t)\}}{\text{min}} & \quad \Xi^\prime(\mathbf p_0(t) + \frac{\mathbf q(t)\odot \mathbf q(t)}{\mathbf h^T(t) (\mathbf q(t) \odot \mathbf q(t))}\\
&\qquad  \cdot (P_\mathrm{max} - \mathbf h^T(t) \mathbf p_0(t)), \mathbf u(t)),
\end{aligned}
\end{equation}
which can be tackled by the classic SGD or Adam optimizer.

Next, based on the principle of AO, the transmitting decision factor can be optimized by the exhaustive search method with the other optimization variables fixed. Additionally, given the user association and ZF transceiver beamforming design, we can employ the exhaustive search algorithm to find the optimal values $L_i(t), i=1,2,\cdots, U$, within the finite set $\mathcal L$, subject to the transmission rate $R_i(t), i=1,2,\cdots, U$, constraint \eqref{p: L-c}.

\subsection{Convergence and Complexity Analysis}
Convergence analysis: The convergence of the objective values is guaranteed by the following three facts. First, the objective value of the initial problem is non-increasing over iterations since each subproblem can achieve the stationary point or optimal solution. Second, the optimal objective value of the problem is guaranteed to converge. Third, since each variable is bounded, thus there must exist a convergent subsequence. 

Complexity analysis: The complexity of updating $\mathbf{V}_{i}$ at each iteration of the subproblem is $\mathcal O(U((N_t L_i)^{4} + N_r^{3.5} + N_r^2  N_t  L_i))$. Similarly, the complexity of updating $\mathbf{U}_{i}$ is $\mathcal O(U((N_t L_i)^{4} + N_r^{3.5} + N_r^2  N_t  L_i))$. The complexity of updating $L_i$ depends on the number of exhaustions $\Gamma$. The complexity of the design of semantic symbols is $\mathcal O(\Gamma(LC_{in}k^2HWC_{out} + C_{in}L_i + UN_r^2N_tL_i + H^2W^2))$, where $L$, $C_{in}$, $C_{out}$,  and $k^2$ are the number of convolutional layers, input channels, output channels, and the size of kernels, respectively. 

\begin{figure*}[htpb]
	 \vspace{-0.35cm}  
	 \setlength{\abovecaptionskip}{0.2cm}   
	 \setlength{\belowcaptionskip}{-2.8cm}   
	\centering
	\subfloat[PSNR for the reconstructed video frames.]{
		\label{PSNR_P}
		\includegraphics[width=0.48\linewidth]{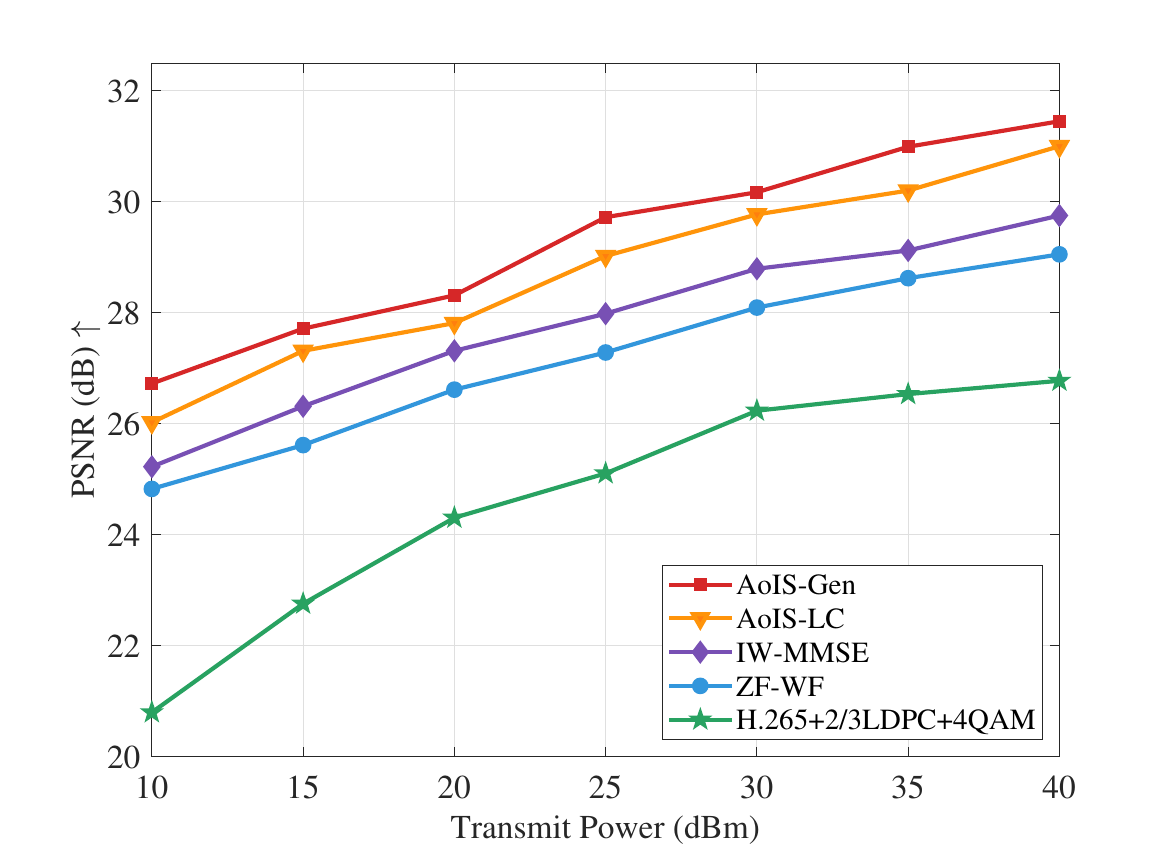}
	}\hfill
	\subfloat[MS-SSIM for the reconstructed video frames.]{
		\label{SSIM_P}
		\includegraphics[width=0.48\linewidth]{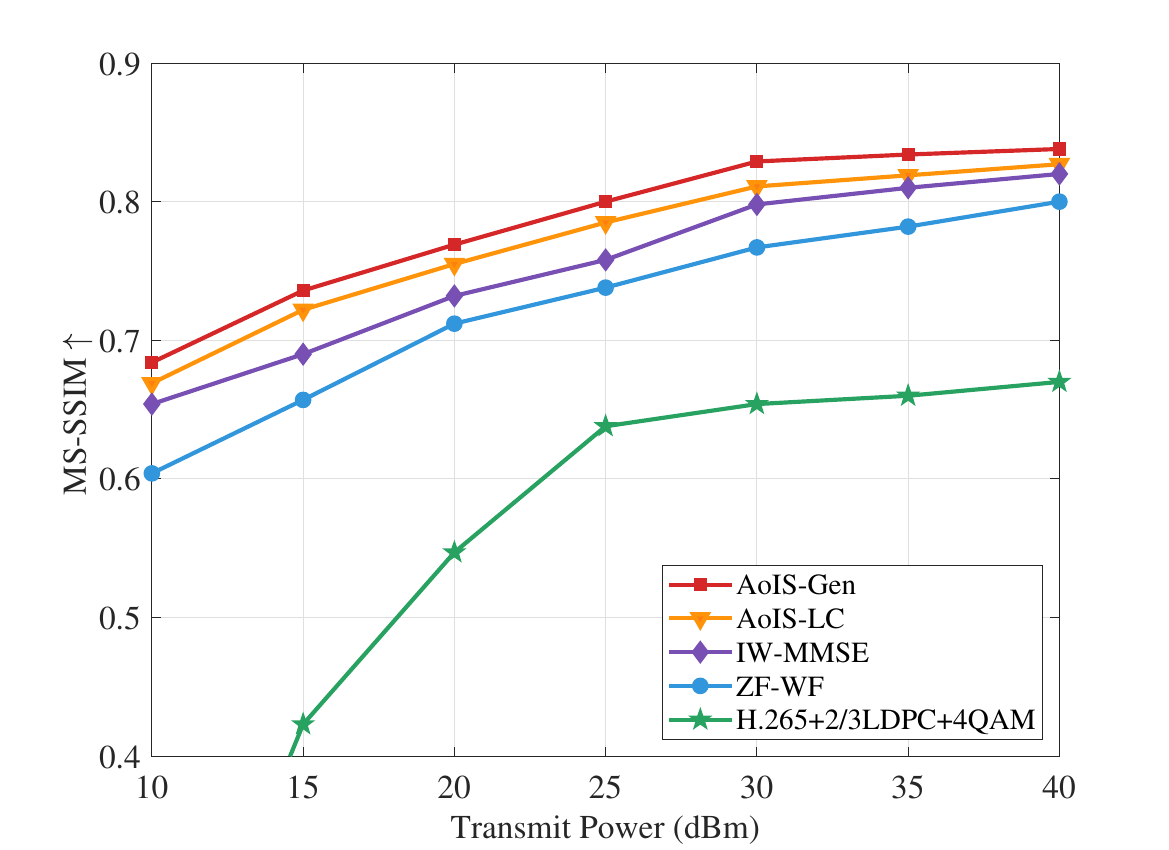}
	}
	\caption{Performance comparison of different solutions versus transmit power on the UCF101 dataset. (a) PSNR performance comparison; (b) MS-SSIM performance comparison.}
	\label{PSNR/SSIM}
\end{figure*}

\section{Numerical Results}
\label{sec:Results}
In this section, numerical results are provided to validate the effectiveness of the proposed semantic-aware resource allocation scheme.

\subsection{Experimental Setups}

\subsubsection{Datasets}
We compare the performance of the proposed algorithms with other benchmarks over the UCF101 \cite{UCF101} dataset, which  consists of 13,320 video clips, with a large variety of scenes and actions. The dataset is split into about a 4:1 ratio for training and testing, respectively. Data augmentation is adopted to increase sample diversity and improve model generalization ability and robustness, which encompass frame random horizontal flips, vertical flips, and random cropping to dimensions of $128 \times 128 \times 3$.

\subsubsection{Parameters and Training Details}
We utilize \textit{Deepwive} \cite{Deepwive} as our backbone network, and the total experiments are conducted on a server with eight RTX 4090 GPU cards. For the model pre-training, the optimizer is Adam \cite{Adam}, and we employ a dynamic learning rate schedule that progressively decreases from $1 \times 10^{-4}$ to $2 \times 10^{-5}$ along with different epochs. The batch size is configured to 1, and the total timeslot is set as 1000. In the video stream, each frame is transmitted by the sender in each time slot.
The number of users is $U=4$. The path loss model is 128.1 + 37.6 $\mathrm {\log [d(km)]}$dB with the standard deviation of shadow fading at 6 dB.  The power spectral density of Gaussian noise is equal to $\sigma^2$ = -174 dBm/Hz. The maximum tolerable delay and transmit power are set to $T_\mathrm{max}$ = 1 ms and $P_\mathrm{max}$ = 50 dBm, respectively. Total channel bandwidth $B$ is 5 MHz. 

\subsubsection{Evaluation Metrics} We leverage the pixel-wise metric peak signal-to-noise ratio (PSNR) and the perceptual-level multi-scale structural similarity (MS-SSIM) as measurements for the reconstructed frame quality. Higher PSNR/MS-SSIM values indicate better performance. To be more closely aligned with the purpose of semantic communications, we adopt the deep learned perceptual image patch similarity (LPIPS) metric \cite{LPIPS} as a perceptual loss measure. LPIPS mimics human perceptual assessment, with lower values indicating less distortion, with a maximum value of 1.

\subsubsection{Comparison Schemes}
We compare the proposed scheme with the following several
benchmarks: 

\begin{itemize}
\item {AoIS-Gen.} In this case, the wireless video transmission backbone is the {{JSCC-based}} network, and the transceiver beamformer is optimized by the general SCA method.

\item {AoIS-LC.} For this proposed scheme, the transceiver beamformer optimization problem is addressed using the proposed low-complexity ZF method. This approach simplifies the optimization process while maintaining effective performance in mitigating interference.

\item {IW-MMSE.} We adopt the scheme proposed in \cite{IW-MMSE}, where the achievable data rate of a system is characterized as a function of the weighted minimum mean square error (WMMSE) matrix of the recovered signals, optimizing the trade-off between rate and distortion.

\item {ZF-WF.} This benchmark incorporates the wireless video transmission framework \textit{Deepwive} which integrates ZF precoding with the water-filling (WF) power allocation strategy to achieve resource optimization.

\item {H.265+LDPC+QAM.} The separated source and channel coding scheme utilizes the H.265 \cite{H.265} video codec for source coding, low-density parity check (LDPC) code for channel coding, quadrature amplitude modulation (QAM) for modulation, and SVD for precoding. In this case, ZF and WF are adopted for channel equalization. 

\end{itemize}

\begin{figure*}
	\centering
	\subfloat{
		\includegraphics[width=0.95\linewidth]{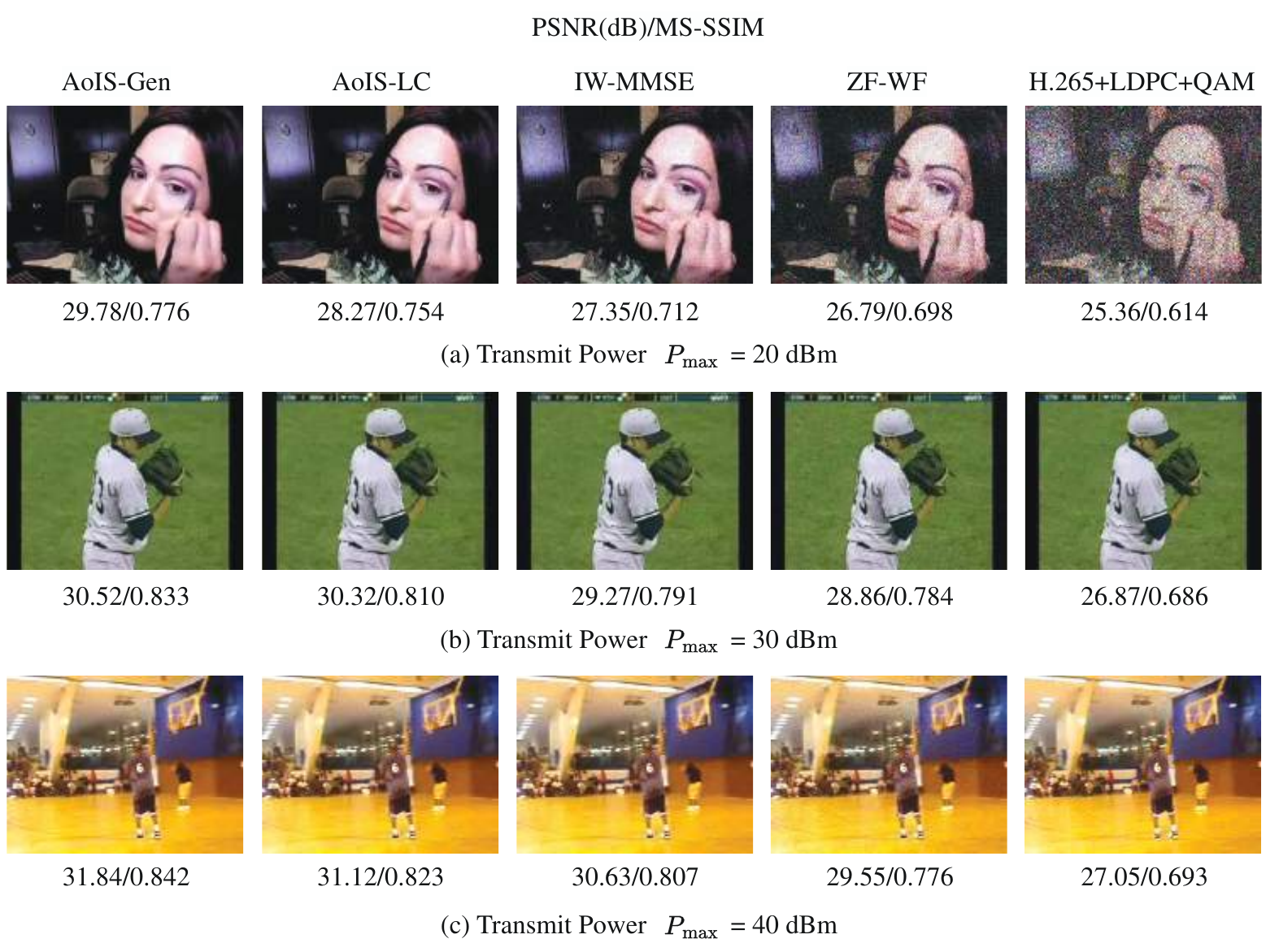}}
	\caption{Visual comparison of reconstructed video frames. The first, second, and third rows correspond to $P_\mathrm{max}$ = 20, 30 and 40 dBm, respectively. The first column shows the results of the AoIS-Gen scheme, and the second one presents the AoIS-LC results, while the third to fifth columns show the reconstructions of the IW-MMSE scheme, ZF-WF, and H.265+3/4LDPC+4QAM, respectively.}
	\label{Visualization}
\end{figure*}


\subsection{Result Analysis}
\subsubsection{Performance for Different Transmit Power}
We first investigate the reconstruction frame quality of our proposed AoIS-aware scheme along with that of the other baseline schemes in a 2 $\times$ 2 MIMO fading channel. To avoid ambiguity and ensure fair comparison, this MIMO channel is configured as two separate MISO channels. As shown in Fig. \ref{PSNR/SSIM}, it is clearly observed that the AoIS-Gen scheme outperforms all the other benchmarks. For the traditional separated coding approach, we adopt the `H.265+2/3LDPC+4QAM' as a benchmark, employing an LDPC code with a 2/3 rate and 4QAM modulation. It can be observed that the traditional approach experiences significant performance degradation under low transmit power conditions. In contrast, the JSCC-based scheme demonstrates superior robustness to low transmit power due to its integrated optimization of source and channel coding.
In terms of PSNR, as shown in Fig. \ref{PSNR/SSIM}\subref{PSNR_P}, the proposed AoIS-Gen and AoIS-LC schemes demonstrate a gain of about 4.2 dB and 4.6 dB over the traditional separate design approach at high transmit power levels (i.e., $P_\mathrm{max} \geq$ 35 dBm), respectively. Both of them outperform the traditional separate design approach by more than 5 dB at low transmit power levels, highlighting their robustness to challenging signal conditions. In addition, the proposed deep AoIS-Gen scheme consistently exhibits a 2 dB improvement in PSNR over all transmit power levels compared to the IW-MMSE scheme, which is mainly attributed to the effective optimization of the physical layer beamforming in our approach.
In terms of perceptual metrics, as shown in Fig. \ref{PSNR/SSIM}\subref{SSIM_P}, the proposed AoIS-aware schemes exhibit significant improvements over the conventional separate scheme. Conventional schemes are inferior to the DL-based counterparts in the low SNR regime because separate video compression is designed mainly for PSNR, and does not consider perceptual quality. The average MS-SSIM performance of the AoIS-Gen scheme surpasses that of the ZF-WF scheme by approximately 6.14\%. Furthermore, the perceptual metrics of the proposed two schemes are consistently better than the IW-MMSE and ZF-WF at all transmit power levels. These results confirm the effectiveness of the proposed AoIS-aware schemes in exploiting semantically aware video frame data and optimized beamforming for improved frame reconstruction performance.

\subsubsection{Visualization Performance for the Wireless Video Transmission}
To intuitively demonstrate the effectiveness of our proposed algorithm, we present the visualization results for the wireless video frame transmission in Fig. \ref{Visualization}. Here we utilize 3/4 rate LDPC and 4QAM for the conventional baseline. It is observed that for different transmit powers, the two proposed algorithms outperform the other baseline schemes in both evaluation metrics and visual quality since JSCC-based neural networks can retain more high-frequency video signal details. Specifically, in low transmit power, such as 20 dBm, IW-MMSE, ZF-WF, and the traditional schemes are unable to avoid blurry points and color noises for reconstructed video frames. Notably, at a high transmit power of 40 dBm, the baseline methods exhibit some distortions, whereas the proposed AoIS-aware scheme achieves better reconstruction with finer details, demonstrating its robustness across diverse power regimes. The visual performance of the frames reconstructed by the proposed AoIS-aware scheme has higher clarity and better detail representation compared to the baseline schemes. This further confirms the fact that the proposed schemes have significant advantages over the other baseline approaches. 

\begin{figure}[t]
	\centering
	\subfloat{
		\includegraphics[width=0.95\linewidth]{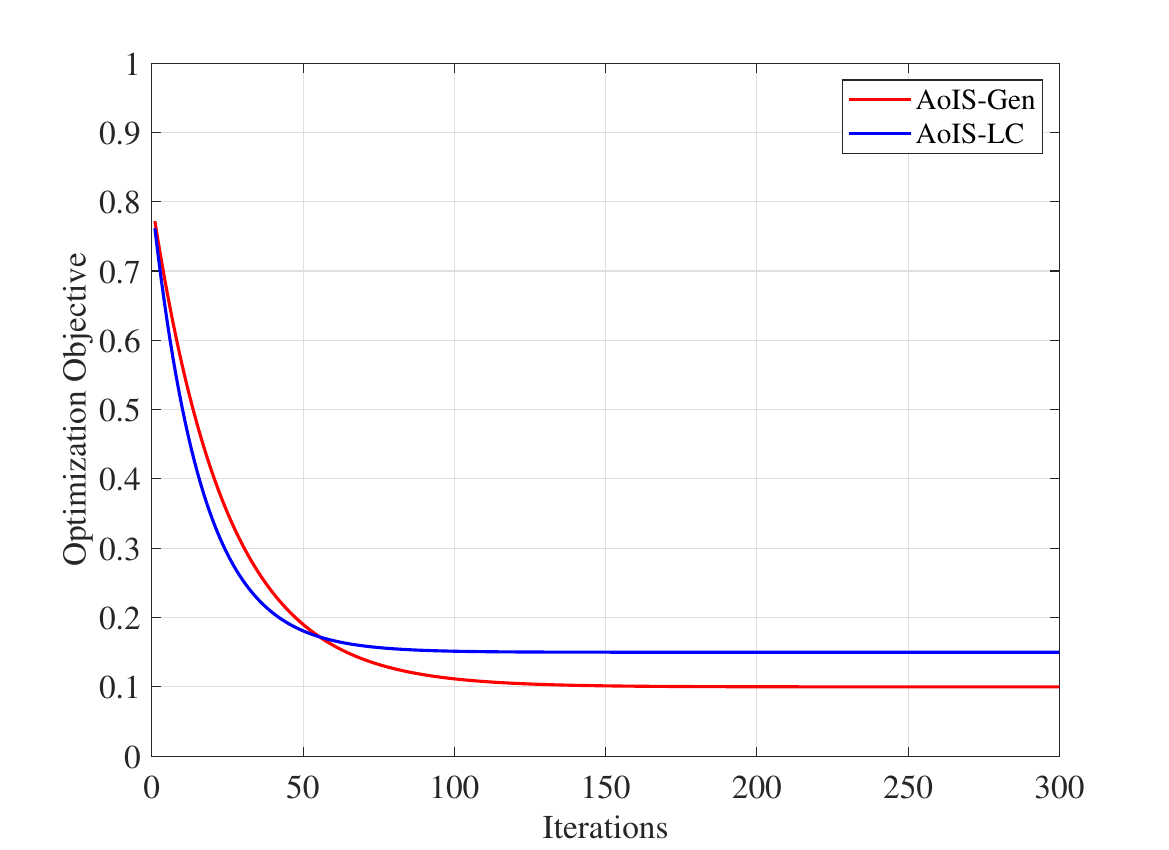}
	}
	\caption{Convergence performance of the proposed AoIS-Gen and AoIS-LC algorithms.}
	\label{Loss}
\end{figure}

\begin{figure}
	\subfloat{
		\includegraphics[width=1\linewidth]{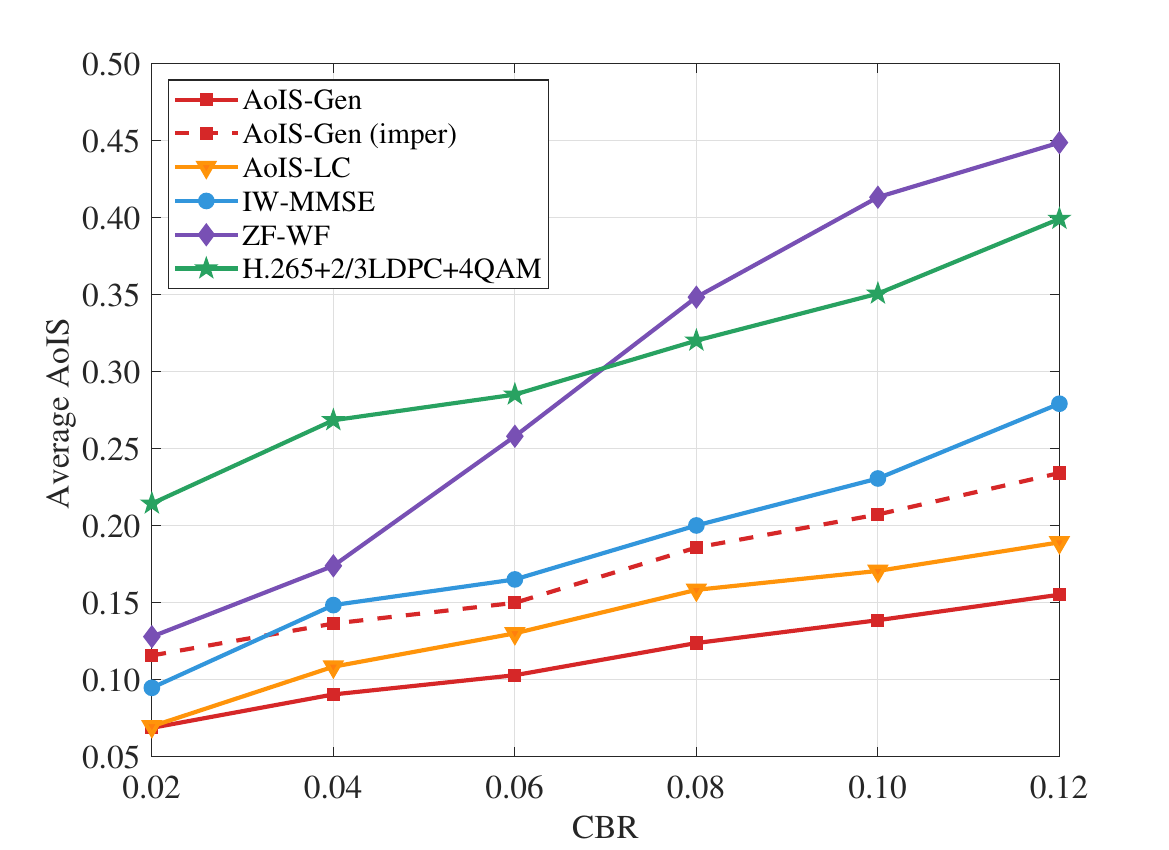}}
	\caption{Average AoIS versus different numbers of sources in a fading channel with $P_\mathrm{max}=35$ dBm.}
	\label{AoIS_CBR}
\end{figure}

\subsubsection{Convergence and Optimality Performance} Fig. \ref{Loss} demonstrates the convergence performance of both the proposed AoIS-Gen and low-complexity AoIS-LC algorithms for resource allocation. It is observed that the loss decreases with an increasing number of iterations and then converges to a stable value, indicating the effectiveness of the AoIS-aware method in finding viable scheduling and resource allocation strategies. Specifically, the AoIS-LC scheme achieves faster convergence compared to the AoIS-Gen scheme. The AoIS-Gen method has a lower average AoIS because the SCA optimization method can effectively optimize resource allocation, thus multiple iterations are required to achieve convergence. The proposed algorithm employs SCA to optimize the transceiver beamformers, achieving local optimality. Meanwhile, the semantic actuation indicator $\alpha_i(t)$ and the semantic symbol length $L_i(t)$ are optimized via exhaustive search, ensuring global optimality for these variables. In this way, the overall solution achieves suboptimal due to the joint optimization of multiple variables.

\begin{figure}[t]
	\centering
	\subfloat{
		\includegraphics[width=1\linewidth]{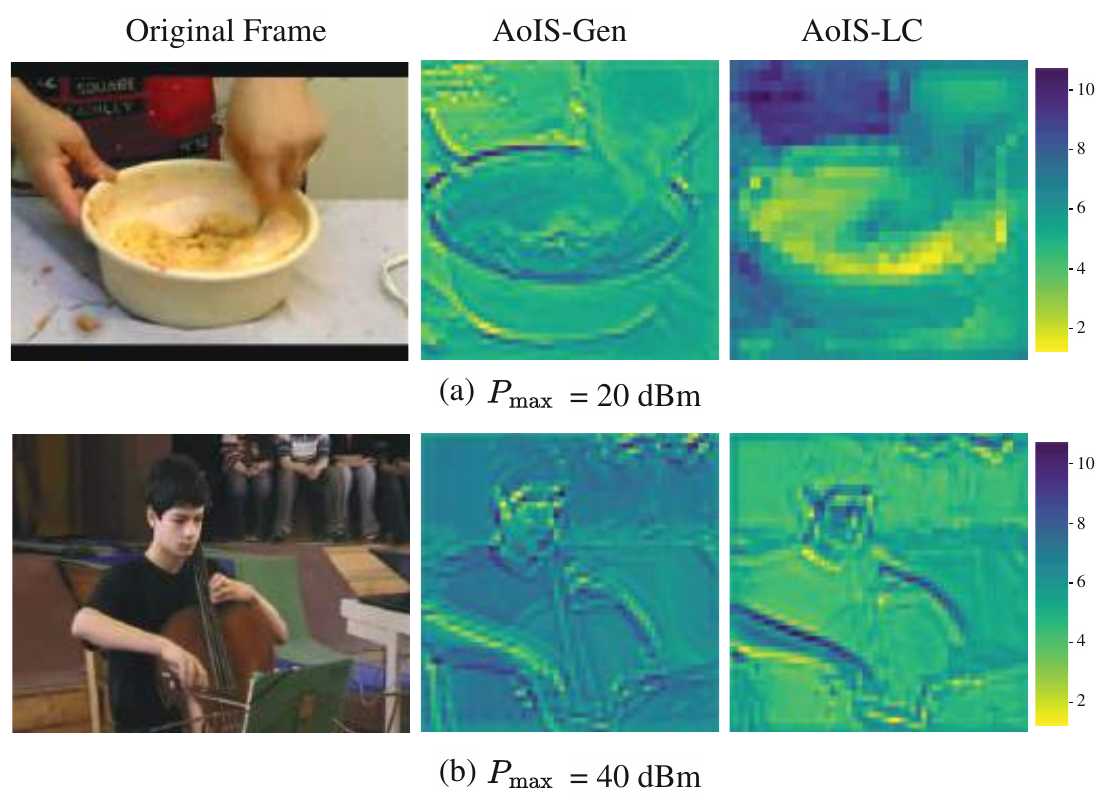}}
	\caption{Visualization of frame resource allocation map for feature vector with transmit power set to 20 dBm and 40 dBm.}
	\label{Visualization_feature}
\end{figure}

\subsubsection{Performance of Average AoIS}
Fig. \ref{AoIS_CBR} demonstrates the impact of different semantic data sizes on the average AoIS for several methods given the transmit power of 35 dBm. 
Moreover, to thoroughly assess the impact of channel estimation error on the system performance, we also conduct simulations under scenarios with imperfect CSI, labeled as ``AoIS-Gen (imper)'', and leverage the MMSE criterion for channel estimation.
We observe that the averaged AoIS performance of the JSCC-based schemes is superior to that of the traditional baseline. The traditional separate baseline scheme lacks a semantic codec, and in this case, the AoIS is reduced to AoI, as semantic information is not considered. 
Typically, a larger semantic symbol length $L_i$ value indicates more preserved semantic information for better reconstruction frame quality and more transmission load.
We observe that the average AoIS of all five methods increases with a larger amount of transmitted source information, since more transmission data causes more intense competition and lower actuation frequency. Among them, the proposed AoIS-aware schemes outperform the others and achieve the best AoIS results, proving the effectiveness.
Moreover, we can observe that the AoIS-Gen (imper) with imperfect channel estimation, slightly inferior to AoIS-Gen with perfect CSI, still notably outperforms the conventional communication methods and the DL-based semantic communication system. This further demonstrates the effectiveness of our system in handling imperfect channel estimation and maintaining superior performance.
For the same average AoIS, our proposed schemes can save more than 50\% semantic information. In the high CBR region, the traditional separate scheme shows better performance compared to the ZF-WF scheme. It can be attributed to the fact that, under constrained communication resources, the complex neural network structures can lead to substantial computational overhead when transmitting large volumes of semantic data.

\subsubsection{Visualization of Resource Allocation Map}
To further investigate the relationship between system communication resource and semantic symbol length, we present the resource allocation map in Fig. \ref{Visualization_feature}. It is observed that the pixels belong to the same object are assigned similar code lengths, whereas the boundaries and contours between distinct objects are allocated more resources for semantic features. Furthermore, as the transmit power increases, a greater number of code lengths are dedicated to capturing the fine details of objects, which enhances the fidelity of the reconstructed video frames.

\section{Conclusion}
\label{sec:Conclusion}
In this paper, we proposed a novel metric called AoIS and investigated the averaged AoIS minimization problem in a remote monitoring system for efficient wireless video transmission. The AoIS minimization problem was solved by the proposed Lyapunov-guided optimization method. Then, we adopted the AO algorithm to optimize all variables, including the semantic actuation indicator, {{transceiver}} beamformer, and semantic symbol length. Specifically, we proposed two methods to optimize the transceiver beamformer, i.e., the SCA method and a low-complexity ZF method. The length of transmitted symbols and the semantic actuation indicator were designed by the {{exhaustive}} searching algorithm. Numerical results demonstrated the validation of the AoIS metric and the effectiveness of the proposed SCA and ZF algorithms, indicating the huge potential of semantic communication in improving the {{efficiency}} of future wireless communication systems.

\appendices
\section{Proof of Proposition \ref{Pro}}
\label{Appendix_Pro}
For any discrete-time queueing system and for any $t > 0$ time slot $0 \leq \tau \leq {t-1}$, from \eqref{Q_{t+1}} we have  
\begin{equation}
\begin{aligned}
\frac{Q_i(t)} {t} - \frac{Q_i(0)} {t} \geq 
& \frac{1} {t} \sum\limits_{\tau=0}^{t-1}\alpha_i(\tau)c_i(\tau) - \frac{1} {t}\sum\limits_{\tau=0}^{t-1}c_\mathrm{max}
\\& =\frac{1} {t} \sum\limits_{\tau=0}^{t-1}\alpha_i(\tau)c_i(\tau) - c_\mathrm{max}.
\end{aligned}
\label{pro: stable}
\end{equation}
Next, taking the limitation of \eqref{pro: stable} yields
\begin{equation}
\lim\limits_{t\to\infty}\frac{Q_i(t)} {t} - \lim\limits_{t\to\infty}\frac{Q_i(0)} {t} \geq 
\lim\limits_{t\to\infty} \frac{1} {t} \sum\limits_{\tau=0}^{t-1}\alpha_i(\tau)c_i(\tau) - c_\mathrm{max}.
\label{pro: lim_stable}
\end{equation}
Since $Q_i(t)$ is stable i.e., $\lim\limits_{t\to\infty}\frac{Q_i(t)} {t}=0$, and $Q_i(0)$ is constant, we have
\begin{equation}
\frac{1}{t}\sum_{\tau=0}^{t-1} \alpha_i(\tau)c_i(\tau) - c_\mathrm{max} \leq 0,
\end{equation}
which implies the constraint $\bar c_i \leq c_\mathrm{max}$ in \eqref{a} is stisfied.

\section{Proof of Proposition \ref{prop: ZF}}
\label{appendix prop: ZF}
In this section, we mainly use the following definitions
\begin{equation}
\begin{aligned}
\mathbf p_0(t) &\triangleq (2^{L_{1}(t)/(BT_\mathrm{max})}-1, \cdots, 2^{L_{U}(t)/(BT_\mathrm{max})}-1)^T,\\
\mathbf p(t) &\triangleq \left(p_1(t), p_2(t), \cdots, p_{U}(t)\right)^T,
\end{aligned}
\end{equation}
to derive the equivalence between the constraints \eqref{p ZF beamforming 1-b}-\eqref{p ZF beamforming 1-c} and Eq. \eqref{eq: constraints equivalence}. Before we start the proof, we first present the following facts. Given that we employ ZF beamforming, all precoders are orthogonal, and increasing power can enhance system performance. Based on this premise, we understand that if $\mathbf{h}^T(t)\mathbf{p}(t) < P_{\text{max}}$, it is possible to identify a better precoder vector $\mathbf{p}(t)$ such that the inner product equals the maximum power, i.e., for every time slot $t$,
\begin{equation}
\label{p perfect: ZF beamforming 1-b equivalence}
\mathbf{h}^T(t)\mathbf{p}(t) = P_{\text{max}}.
\end{equation}
Hence, in the subsequent, we need to prove the equivalence between the two equations, namely~\eqref{p ZF beamforming 1-c} and~\eqref{p perfect: ZF beamforming 1-b equivalence}, and the equation~\eqref{eq: constraints equivalence}.

\textit{Necessity:} From Eq.~\eqref{p ZF beamforming 1-c}, we can introduce an auxiliary variable $\mathbf q(t)=(q_1(t), q_2(t),\cdots,q_U(t))^T$ such that
\begin{equation}
p_i(t) = 2^{L_{i}(t)/(BT_\mathrm{max})}-1 + q_i^2(t), i=1,2\cdots,U.
\end{equation}
Then, Eq.~\eqref{p ZF beamforming 1-c} is equivalent to
\begin{equation}
\begin{aligned}
\mathbf p(t) &= \mathbf p_{0}(t) + \mathbf q(t) \odot \mathbf q(t),
\end{aligned}
\end{equation}
which can be substituted into \eqref{p perfect: ZF beamforming 1-b equivalence} and we get
\begin{equation}
\begin{aligned}
&\mathbf{h}^T(t)(\mathbf p_{0}(t) + \mathbf q(t) \odot \mathbf q(t)) = P_{\mathrm{max}},\\
&\frac{P_{\mathrm{max}} - \mathbf{h}^T(t)\mathbf p_{0}(t)}{\mathbf{h}^T(t)(\mathbf q(t) \odot \mathbf q(t))}=1.
\end{aligned}
\end{equation}
Combining the above equations, we have
\begin{equation}
\begin{aligned}
\mathbf p(t) &= \mathbf p_0(t) + \frac{\mathbf q(t) \odot \mathbf q(t)}{\mathbf h^T(t) (\mathbf q(t) \odot \mathbf q(t))}(P_\mathrm{max} - \mathbf h^T(t) \mathbf p_0(t)).
\end{aligned}
\end{equation}

\textit{Sufficiency:} Since $P_\mathrm{max}\geq \mathbf h^T(t) \mathbf p_0(t)$, from \eqref{eq: constraints equivalence} we have
\begin{equation}
\begin{aligned}
&\mathbf h^T(t)\mathbf p(t) = \mathbf h^T(t) \bigg(\mathbf p_0(t) + \frac{\mathbf q(t) \odot \mathbf q(t)}{\mathbf h^T(t) (\mathbf q(t) \odot \mathbf q(t))} \\
&\qquad \qquad \qquad \qquad  \cdot (P_\mathrm{max} - \mathbf h^T(t) \mathbf p_0(t))\bigg) =P_\mathrm{max},
\end{aligned}
\end{equation}
and
\begin{equation}
\begin{aligned}
\mathbf p(t) 
& = \mathbf p_0(t) + \frac{\mathbf q(t) \odot \mathbf q(t)}{\mathbf h^T(t) (\mathbf q(t) \odot \mathbf q(t))}(P_\mathrm{max} - \mathbf h^T(t) \mathbf p_0(t)) \\
&\geq \mathbf p_0(t),
\end{aligned}
\end{equation}
From the definitions $\mathbf p_0(t)$ and $\mathbf p(t)$, we know that for each user $i$,
\begin{equation}
p_i(t) \geq 2^{L_{i}(t)/(BT_\mathrm{max})}-1,\forall i.
\end{equation}
Hence, the constraints \eqref{p ZF beamforming 1-b}-\eqref{p ZF beamforming 1-c} hold.

\end{document}